\UseRawInputEncoding
\documentclass[12pt, centerh1]{article}
\usepackage{natbib}
\usepackage[margin=1in]{geometry}
\usepackage{amsmath,natbib}
\usepackage{amsfonts,mathrsfs,moreverb,lipsum}
\usepackage{graphicx,colonequals}
\usepackage{fixltx2e}
\usepackage{color}
\usepackage{caption}
\usepackage{subcaption}
\usepackage{yfonts}
\usepackage{pbox}
\usepackage{amsthm}
\usepackage[mathscr]{eucal}
\usepackage{float}
\usepackage{authblk}

\newtheorem{theorem}{Theorem}[section]
\newtheorem{corollary}{Corollary}[section]

\allowdisplaybreaks

\newcommand{\btheta}{{\boldsymbol{\theta}}}
\newcommand{\bvtheta}{{\boldsymbol{\vartheta}}}

\newcommand\numberthis{\addtocounter{equation}{1}\tag{\theequation}}

\newcommand{\vecmu}{\mbox{\boldmath$\mu$}}

\newcommand{\vecw}{\mathbf{w}}

\newcommand{\vect}{\mathbf{t}}
\newcommand{\vecu}{\mathbf{u}}
\newcommand{\vecX}{\mathbf{X}}

\newcommand{\vecV}{\mathbf{V}}

\newcommand{\vecz}{\mathbf{z}}

\newcommand{\vecZ}{\mathbf{Z}}

\newcommand{\matdel}{\mathbf\Delta}
\newcommand{\matSig}{\mathbf\Sigma}
\newcommand{\matPsi}{\mathbf\Psi}

\newcommand{\tr}{\,\mbox{tr}}

\newcommand{\vecA}{\mathbf{A}}
\newcommand{\vecE}{\mathbf{E}}

\newcommand{\matm}{\mathbf{M}}

\newcommand{\tenM}{\textgoth{M}}
\newcommand{\tenA}{\textgoth{A}}
\newcommand{\tenX}{\textgoth{X}}
\newcommand{\tenV}{\textgoth{V}}
\newcommand{\tenT}{\textgoth{T}}
\newcommand{\ident}{\mathbf{I}}

\newcommand{\vecalpha}{\mbox{\boldmath$\alpha$}}

\newcommand{\vecDelta}{\mbox{\boldmath$\Delta$}}

\newcommand{\vecalp}{\boldsymbol{\alpha}}
\newcommand{\fX}{\mathscr{X}}
\newcommand{\fV}{\mathscr{V}}
\newcommand{\fZ}{\mathscr{Z}}

\newcommand{\vecc}{\text{vec}}

\DeclareMathOperator{\EV}{\mathbb{E}}

\newcommand{\half}{\frac{1}{2}}
\newcommand{\mhalf}{-\frac{1}{2}}
\newcommand{\T}{^{\top}}
\newcommand{\inv}{^{-1}}

\newcommand{\vecn}{\mathbf{n}}
\newcommand{\vecM}{\mathbf{M}}
\newcommand{\vecMone}{\mathbf{M}_{(1)}}
\newcommand{\cholDeltaI}[1]{\vecDelta_{#1}^{\mhalf}}

\newcommand{\norm}[1]{\left\lVert#1\right\rVert}

\newcommand{\vecej}{\mathbf{e}_{j}}
\newcommand{\vecejT}{\mathbf{e}_{j}\T}

\title{Four Skewed Tensor Distributions}
\author[1]{Michael P.B. Gallaugher\thanks{$^*$Corresponding author. Email: Michael\_Gallaugher@baylor.edu}}
\author[2]{Peter A. Tait}
\author[2]{Paul D. McNicholas}
\affil[1]{Department of Statistical Science, Baylor University, Waco, Texas, USA}
\affil[2]{Department of Mathematics and Statistics, McMaster University, Ontario, Canada}

\date{}

\begin{document}

\maketitle{}

\begin{abstract}
With the rise of the ``big data" phenomenon in recent years, data is coming in many different complex forms. One example of this is multi-way data that come in the form of higher-order tensors such as coloured images and movie clips. Although there has been a recent rise in models for looking at the simple case of three-way data in the form of matrices, there is a relative paucity of higher-order tensor variate methods. The most common tensor distribution in the literature is the tensor variate normal distribution; however, its use can be problematic if the data exhibit skewness or outliers. Herein, we develop four skewed tensor variate distributions which to our knowledge are the first skewed tensor distributions to be proposed in the literature, and are able to parameterize both skewness and tail weight. Properties and parameter estimation are discussed, and real and simulated data are used for illustration.
\end{abstract}

\section{Introduction}
In the last decade, data is coming in increasingly complex structures, and therefore traditional statistical methods are often either not ideal or not applicable. One such complex structure is multiway or tensor type data. The simplest of these forms is three-way data which come in the form of matrices or order-2 tensors. In the last few years, there have been numerous examples of modelling three-way data and specifically skewed three-way data, including the development of four skewed matrix variate distributions \citep{gallaugher17,gallaugher19} as well as the use of these in a mixture-model setting for clustering and classification \citep{gallaugher18a, gallaugher19b}. Transformation methods have also been applied to three-way data such as the work by \cite{melnykov18}. Examples of three-way data include multivariable longitudinal data as well as greyscale images. 

Although these aforementioned methods are useful, they are nevertheless restricted to three-way data. More interesting data types such as coloured images and movie clips (black and white or coloured) come in the form of multilinear data or order-$D$ tensors. For example, black and white movie clips consist of greyscale images (matrices) collected at different time points, and therefore would come in the form of an order-$3$ pixel intensity tensor. In the case of coloured images, the data would again come in the form of an order-$3$ tensor with pixel intensity matrices for, generally three, different colours. Finally, coloured movie clips would come in the form of a order-$4$ tensor represented as a hyper-cuboid of pixel intensities. 

Currently, to our knowledge, analysis of such tensor type data is restricted to the multilinear/tensor variate normal (TVN) distribution and was used very recently in the area of clustering and classification \citep{tait2020}. Although mathematically tractable, the assumption of symmetry is often violated. Moreover, outliers may be present in the data which can be problematic. To fill this gap, we present four skewed tensor variate distributions which can be considered generalizations of their matrix variate counterparts, and are able to model both skewness and excess kurtosis.

An outline of our contributions is now presented. We first present a detailed derivation of a tensor variate skew-$t$ (TVST) distribution via a tensor extension of the normal variance mean mixture model. Three other tensor distributions also fall naturally out of this derivation, namely the tensor variate generalized hyperbolic (TVGH), variance gamma (TVVG), and normal inverse Gaussian (TVNIG) distributions. Furthermore, a tensor variate shifted asymmetric Laplace (TVSAL) distribution comes out as a special case of the TVVG distribution. Properties of these four distributions including expectation, characteristic functions, matricizations, and relationships to other distributions are then discussed. Two avenues are explored for parameter estimation, both using an expectation conditional maximization algorithm. The first, although more mathematically tractable, can be computationally inefficient as it requires the matricization of the tensor along each of its modes. The second method proposed makes use of only mode-one matricizations and permutation operators, and is far more computationally feasible. Finally, the proposed distributions are fit to colour images in the form of order-3 tensors. We finish with a discussion and possible paths for future work that include incorporating these distributions in the mixture model setting for use in clustering and classification, as well as dimension reduction techniques.

The remainder of this paper is laid out as follows. Section~2 presents a detailed background on the tensor variate normal distribution as well as the inverse and generalized inverse Gaussian distributions which will be used in the formulation of the four skewed tensor distributions. In Section~3 the four skewed tensor distributions are derived and their properties discussed. Two parameter estimation procedures are discussed in Section~4. Simulation and read data analyses are presented in Section~5, and we finish with a discussion and avenues for future work in Section~6.

\section{Background}

\subsection{Tensor Variate Normal Distribution}
As with the univariate, multivariate, and matrix variate cases, the TVN distribution is the most well-known tensor variate distribution, and its form, mathematical properties, and parameter estimation are thoroughly discussed in the literature. 

If $\mathcal{X}$ is a random order-$D$ tensor, with dimensional lengths $n_1\times n_2\times \cdots \times n_D = \vecn$ , with realization $\tenX$, then it follows that the density function of a TVN distribution, $\mathcal{N}_\vecn\left(\tenM,\bigotimes_{d=1}^D\matdel_d\right)$, can be written as
\begin{equation}
\begin{split}
f(\tenX|\tenM, \matdel_1,\ldots,\matdel_D&)=\\&(2\pi)^{\frac{-n^{*}}{2}}\prod_{d=1}^D|\matdel_d|^{-\frac{n^{*}}{2n_d}} \exp\left\{-\frac{1}{2}\vecc(\tenX-\tenM)\T\bigotimes_{d=1}^D\matdel_d^{-1}\vecc(\tenX-\tenM)\right\},
\end{split}
\label{eq:tensnorm}
\end{equation}
where $\tenM$ is the mean tensor, $\vecc(\cdot)$ is the tensor vectorization operator, $n^{*} = \prod_{d=1}^Dn_d$, and 
$$\bigotimes_{d=1}^D\matdel_d=\matdel_1\otimes\matdel_2\otimes\cdots\otimes\matdel_D,$$
where $\otimes$ represents the Kronecker product \citep{ohlson2013}. Note that $\text{Cov}(\vecc(\mathcal{X}))=\bigotimes_{d=1}^D\matdel_d$, and, for order-3 tensors, we consider the first mode to be the rows, the second mode to be the columns, and the third mode the slices.

One important property that we mention here, and is shown by \cite{ohlson2013}, is that the exponent in the density function \eqref{eq:tensnorm} can be written as
\begin{equation}
-\frac{1}{2}\tr[\matdel_j^{-1}(\vecX_{(j)}-\matm_{(j)})\T\bigotimes_{d\ne j}\matdel_{d}^{-1}(\vecX_{(j)}-\matm_{(j)})],
\label{eq:exp}
\end{equation}
for $1\le j \le D$, where $\vecV_{(j)}$ is the matricization along mode $j$ for a tensor $\tenV$.
Additionally, if $\mathcal{X}$ is a $q\times r$ random matrix then  
\begin{equation}
\mathcal{X}\sim \mathcal{N}_{q\times r}(\matm,\matSig,\matPsi)\iff\vecc(\mathcal{X})\sim \phi_{qr}(\vecc(\matm),\matPsi\otimes\matSig) 
\label{eq:MVMV},
\end{equation}   
where $\mathcal{N}_{q\times r}(\cdot)$ represents the matrix variate normal distribution with mean matrix $\matm$, $q\times q$ row covariance matrix $\matSig$, and $r\times r$ column covariance matrix $\matPsi$ and $\phi_{qr}(\cdot)$ represents the multivariate normal distribution of dimension $np$. Between \eqref{eq:exp} and \eqref{eq:MVMV}, we easily arrive at Theorem \ref{the:equiv}.

\begin{theorem}
If $\fX$ is an order-$D$ random tensor of dimension $\vecn$ then the following statements are equivalent.
\begin{enumerate}
\item $\fX\sim \mathcal{N}_{\vecn}\left(\tenM,\bigotimes_{d=1}^D\matdel_d\right)$
\item $\mathcal{X}_{(j)}\sim \mathcal{N}_{\frac{n^*}{n_j}\times n_j}\left(\matm_{(j)},\bigotimes_{d\ne j}\matdel_d, \matdel_j\right)$
\item $\vecc(\mathcal{X}_{(j)})\sim \phi_{n^*}\left(\vecc({\matm_{(j)}}),\matdel_j\otimes\bigotimes_{d\ne j}\matdel_d\right)$
\end{enumerate}
\label{the:equiv}
\end{theorem}
Regarding notation, we will use $\mathcal{X}_{(j)}$ to represent the matricization along mode $j$ of a random tensor $\fX$.
Due to the relationship between the tensor variate and multivariate normal distribution, it is simple to derive the characteristic function as discussed in \cite{ohlson2013}. If $\fX\sim\mathcal{N}_{\vecn}\left(\tenM,\bigotimes_{d=1}^D\matdel_d\right)$ and $\tenT$ is an $\vecn$ dimensional order-$D$ tensor, then the characteristic function of $\fX$ is given by
\begin{equation}
\mathcal{C}(\tenT)\colonequals\mathbb{E}[\exp\{i\vecc(\tenT)'\vecc(\fX)\}]=\exp\left\{i\vecc(\tenT)'\vecc(\tenM)-\frac{1}{2}\vecc(\tenT)'\bigotimes_{d=1}^D\matdel_d\vecc(\tenT)\right\}.
\label{eq:CFN}
\end{equation}

\subsection{Inverse and Generalized Inverse Gaussian Distribution}
The derivation of the TVGH and TVNIG distributions, as well as parameter estimation for all four skewed tensor variate distributions, will rely on the inverse and generalized inverse Gaussian distributions. 
A random variable $Y$ has an inverse Gaussian distribution if its probability density function can be written
$$
f(y~|~\delta,\gamma)=\frac{\delta}{\sqrt{2\pi}}\exp\{\delta\gamma\}y^{-\frac{3}{2}}\exp\left\{-\frac{1}{2}\left(\frac{\delta^2}{y}+\gamma^2y\right)\right\},
$$
for $y>0$ and $\delta,\gamma>0$. We denote this distribution by $\text{IG}(\delta,\gamma)$. In the development of the tensor variate NIG distribution, we will consider the standard case where $\delta=1$.

A random variable $Y$ has a generalized inverse Gaussian (GIG) distribution with parameters $a, b$ and $\lambda$ if its density function can be written as
$$
f(y|a, b, \lambda)=\frac{\left(\frac{a}{b}\right)^{\frac{\lambda}{2}}y^{\lambda-1}}{2K_{\lambda}(\sqrt{ab})}\exp\left\{-\frac{ay+\frac{b}{y}}{2}\right\},
$$
where
$$
K_{\lambda}(x)=\frac{1}{2}\int_{0}^{\infty}y^{\lambda-1}\exp\left\{-\frac{x}{2}\left(y+\frac{1}{y}\right)\right\}dy
$$
is the modified Bessel function of the third kind with index $\lambda$.
Finally, the characteristic function of the $\text{GIG}(a,b,\lambda)$ distribution is given by
$$
\mathcal{C}_{\text{GIG}}(t~|~a,b,\lambda)=\left(\frac{a}{a-2it}\right)^{\frac{\lambda}{2}}\frac{K_{\lambda}(\sqrt{b(a-2it)})}{K_{\lambda}(\sqrt{ab})}.
$$
An alternative parameterization for the GIG distribution proposed by \cite{browne15} and used to develop the generalized hyperbolic distribution, will be used to develop the tensor variate generalized hyperbolic. The density using this parameterization is 
\begin{equation}
g(y|\omega,\eta,\lambda)= \frac{\left({y}/{\eta}\right)^{\lambda-1}}{2\eta K_{\lambda}(\omega)}\exp\left\{-\frac{\omega}{2}\left(\frac{y}{\eta}+\frac{\eta}{y}\right)\right\},
\label{eq:I}
\end{equation}
where $\omega=\sqrt{ab}$ and $\eta=\sqrt{a/b}$. For notational clarity, we will denote the parameterization given in \eqref{eq:I} by $\text{I}(\omega,\eta,\lambda)$.

\section{Methodology}
\subsection{Four Skewed Tensor Variate Distributions}
In the multivariate and matrix variate cases, the normal variance mean mixture model is a computationally efficient way to introduce skewness. In the multivariate case, this formulation assumes that the random vector $\vecX$ can be written in the form
$$
\vecX=\vecmu+W\vecalp+\sqrt{W}\vecV,
$$
where $\vecmu$ is a location vector, $\vecalp$ is a skewness vector, $\vecV\sim\phi({\bf 0},\boldsymbol{\Sigma})$, and $W>0$ is a positive random variable. This was extended to the matrix variate case by \cite{gallaugher17,gallaugher19} and this is now extended to the tensor case to derive four skewed tensor variate distributions.

We show below, a derivation of a TVST distribution. Further details are available in Appendix A. The TVGH, TVVG (and TVSAL) and TVNIG distributions can be derived in much the same way and we therefore only give their densities.

We will say that an $\vecn$ order-$D$ random tensor $\fX$ has a tensor variate skew $t$ distribution, $\text{TVST}_{{\bf n}}(\tenM,\tenA,\bigotimes_{d=1}^D\matdel_d,\nu)$ if $\fX$ can be written as
\begin{equation}
\fX=\tenM+W\tenA+\sqrt{W}\fV,
\label{eq:VMM}
\end{equation}
where $\tenM$ and $\tenA$ are $\vecn$ dimensional order-$D$ tensors, $\fV \sim \mathcal{N}_{{\bf n}}\left(\mathbb{O}, \bigotimes_{d=1}^D\matdel_d\right)$ and $W\sim \text{Inv-Gamma}\left(\frac{\nu}{2},\frac{\nu}{2}\right)$ where $\text{Inv-Gamma}(\cdot)$ represents the inverse-gamma distribution. Similar to its multivariate \citep{murray14b} and matrix variate \citep{gallaugher17} counterparts, $\tenM$ is a location tensor, $\tenA$ is the skewness tensor, $\matdel_1,\ldots,\matdel_D$ are scale matrices, and $\nu$ is the degrees of freedom.
It then follows that
$$
\fX|W=w\sim \mathcal{N}_{\bf n}\left(\tenM+w\tenA,w\bigotimes_{d=1}^D\matdel_d\right)
$$
and thus the joint density of $\fX$ and $W$ is
\begin{align*}
f(\tenX,w|\bvtheta)&=f(\tenX|W=w)f(w)\\
&= \frac{\frac{\nu}{2}^{\frac{\nu}{2}}}{(2\pi)^{\frac{n^*}{2}}\prod_{d=1}^D| \matdel_d |^{\frac{n^*}{2n_d}}\Gamma(\frac{\nu}{2})}w^{-\frac{\nu+n^*}{2}-1}\\
& \hspace{0.2in} \times \exp\left\{-\frac{1}{2w}\left(\vecc(\tenX-\tenM-w\tenA)\T\bigotimes_{d=1}^D\matdel_{d}^{-1}\vecc(\tenX-\tenM-w\tenA)+\nu\right)\right\}, \numberthis \label{eqn:joint}
\end{align*}
where $\bvtheta=(\tenM,\tenA,\matdel_1,\ldots,\matdel_D,\nu)$.
%
We note that the exponential term in \eqref{eqn:joint} can be written
\begin{equation*}
\exp\left\{\left(\vecc(\tenX-\tenM)\T\bigotimes_{d=1}^D\matdel_{d}^{-1}\vecc(\tenA)\right)\right\} \exp\left\{-\frac{1}{2}\left[\frac{\delta(\tenX;\tenM,\bigotimes_{d=1}^D\matdel_{d}^{-1})+\nu}{w}+w\rho(\tenA;\bigotimes_{d=1}^D\matdel_{d}^{-1})\right]\right\},
\end{equation*}
where
$\delta(\cdot)=\vecc(\tenX-\tenM)\T\bigotimes_{d=1}^D\matdel_{d}^{-1}\vecc(\tenX-\tenM)$ and 
$\rho(\cdot)=\vecc(\tenA)\T\bigotimes_{d=1}^D\matdel_{d}^{-1}\vecc(\tenA)$.

Therefore, the marginal density of $\tenX$ is
\begin{align*}
f(\tenX)&=\int_{0}^{\infty}f(\tenX,w) dw\\
&=\frac{\frac{\nu}{2}^{\frac{\nu}{2}}}{(2\pi)^{\frac{n^*}{2}}\prod_{d=1}^D| \matdel_d |^{\frac{n^*}{2n_d}} \Gamma(\frac{\nu}{2})}\exp\left\{\vecc(\tenX-\tenM)\T\bigotimes_{d=1}^D\matdel_{d}^{-1}\vecc(\tenA)\right\} \\
& \hspace{0.25in} \times \int_{0}^\infty w^{-\frac{\nu+np}{2}-1}\exp\left\{-\frac{1}{2}\left[\frac{\delta(\tenX;\tenM,\bigotimes_{d=1}^D\matdel_{d}^{-1})+\nu}{w}+w\rho(\tenA,\bigotimes_{d=1}^D\matdel_{d}^{-1})\right]\right\}dw
\end{align*}
Making the change of variables
$$
y=\frac{\sqrt{\rho(\tenA,\bigotimes_{d=1}^D\matdel_{d}^{-1})}}{\sqrt{\delta(\tenX;\tenM,\bigotimes_{d=1}^D\matdel_{d}^{-1})+\nu}}w,
$$
we arrive at the density
\begin{equation}
    \label{eqn:marg}
\begin{split}
\nonumber f_{\text{TVST}}(\tenX|\bvtheta)=&\frac{2\left(\frac{\nu}{2}\right)^{\frac{\nu}{2}}\exp\left\{\vecc(\tenX-\tenM)\T\bigotimes_{d=1}^D\matdel_{d}^{-1}\vecc(\tenA) \right\} }
{(2\pi)^{\frac{n^*}{2}}\prod_{d=1}^D| \matdel_d |^{\frac{n^*}{2n_d}} \Gamma(\frac{\nu}{2})}
  \left(\frac{\delta(\tenX;\tenM,\bigotimes_{d=1}^D\matdel_{d}^{-1})+\nu}{\rho(\vecA,\bigotimes_{d=1}^D\matdel_{d}^{-1})}\right)^{-\frac{\nu+n^*}{4}} \\
  & \times K_{-\frac{\nu+n^*}{2}}\left(\sqrt{\left[\rho(\tenA,\bigotimes_{d=1}^D\matdel_{d}^{-1}))\right]\left[\delta(\tenX;\tenM,\bigotimes_{d=1}^D\matdel_{d}^{-1})+\nu\right]}\right)
  \end{split}
\end{equation}
for $\nu\in\mathbb{R}^+$. For notational purposes we will denote this distribution by $\text{TVST}(\tenM,\tenA,\bigotimes_{d=1}^D\matdel_d,\nu)$.

The density of $\fX$, as derived here, closely resembles, and can be considered a multilinear extension of, the density of the multivariate skew-$t$ distribution given in \cite{murray14b} and the matrix skew-$t$ distribution of \cite{gallaugher17}.
The TVGH, TVVG, and TVNIG distributions are derived in much the same way using the same distributions for $W$ as used for their matrix variate counterparts \citep{gallaugher19}. Specifically, the TVGH distribution arises with $W\sim I(\omega,1,\lambda)$. Its density is given by
\begin{equation}
 \begin{split}
f_{\text{TVGH}}(\tenX|\bvtheta)=&\frac{\exp\left\{\vecc(\tenX-\tenM)\T\bigotimes_{d=1}^D\matdel_{d}^{-1}\vecc(\tenA) \right\} }
{(2\pi)^{\frac{n^*}{2}}\prod_{d=1}^D| \matdel_d |^{^{\frac{n^*}{2n_d}}} K_{\lambda}(\omega)}
  \left(\frac{\delta(\tenX;\tenM,\bigotimes_{d=1}^D\matdel_{d}^{-1})+\omega}{\rho(\vecA,\bigotimes_{d=1}^D\matdel_{d}^{-1})+\omega}\right)^{\frac{\lambda-\frac{n^*}{2}}{2}} \\ & \times
 K_{\lambda-n^*/2}\left(\sqrt{\left[\rho(\tenA,\bigotimes_{d=1}^D\matdel_{d}^{-1})+\omega)\right]\left[\delta(\tenX;\tenM,\bigotimes_{d=1}^D\matdel_{d}^{-1})+\omega\right]}\right)
\end{split}
\label{eq:GH}
\end{equation}
for $\lambda\in\mathbb{R}$, $\omega\in\mathcal{R}^+$. We will denote the tensor variate generalized hyperbolic distribution by $\text{TVGH}(\tenM,\tenA,\bigotimes_{d=1}^D\matdel_d,\lambda,\omega)$. This form is again similar to its multivariate \citep{browne15} and matrix variate \citep{gallaugher19} counterparts.

The TVVG distribution can be derived with $W\sim \text{Gamma}(\gamma,\gamma)$ and the resulting density is
\begin{equation}
 \begin{split}
f_{\text{TVVG}}(\tenX|\bvtheta)=&\frac{2\gamma^{\gamma}\exp\left\{\vecc(\tenX-\tenM)\T\bigotimes_{d=1}^D\matdel_{d}^{-1}\vecc(\tenA) \right\} }
{(2\pi)^{\frac{n^*}{2}}\prod_{d=1}^D| \matdel_d |^{^{\frac{n^*}{2n_d}}} \Gamma(\gamma)}
  \left(\frac{\delta(\tenX;\tenM,\bigotimes_{d=1}^D\matdel_{d}^{-1})}{\rho(\vecA,\bigotimes_{d=1}^D\matdel_{d}^{-1})+2\gamma}\right)^{\frac{\gamma-\frac{n^*}{2}}{2}} \\ & \times
 K_{\gamma-n^*/2}\left(\sqrt{\left[\rho(\tenA,\bigotimes_{d=1}^D\matdel_{d}^{-1})+2\gamma\right]\left[\delta(\tenX;\tenM,\bigotimes_{d=1}^D\matdel_{d}^{-1})\right]}\right),
\end{split}
\end{equation}
where $\gamma\in\mathbb{R}^+$. We will denote this distribution by $\text{TVVG}(\tenM,\tenA,\bigotimes_{d=1}^D\matdel_d,\gamma)$.

Finally, a TVNIG distribution can be derived with $W\sim \text{IG}(1,\kappa)$. Note that the standard form of the inverse Gaussian distribution is used here to allow for the determinants of the scale matrices to be unconstrained \citep{karlis09}. The resulting density function given by
\begin{equation}
 \begin{split}
f_{\text{TVNIG}}(\tenX|\bvtheta)&=\frac{2\exp\left\{\vecc(\tenX-\tenM)\T\bigotimes_{d=1}^D\matdel_{d}^{-1}\vecc(\tenA)+\kappa \right\} }
{(2\pi)^{\frac{n^*+1}{2}}\prod_{d=1}^D| \matdel_d |^{\frac{n^*}{2n_d}}}
  \left(\frac{\delta(\tenX;\tenM,\bigotimes_{d=1}^D\matdel_{d}^{-1})+1}{\rho(\vecA,\bigotimes_{d=1}^D\matdel_{d}^{-1})+\kappa^2}\right)^{-\frac{1+n^*}{4}} \\&\times
 K_{-\frac{1+n^*}{2}}\left(\sqrt{\left[\rho(\tenA,\bigotimes_{d=1}^D\matdel_{d}^{-1})+\kappa^2\right]\left[\delta(\tenX;\tenM,\bigotimes_{d=1}^D\matdel_{d}^{-1})+1\right]}\right),
\end{split}
\label{eq:NIG}
\end{equation}
where $\kappa\in\mathbb{R}^+$. We will use the notation $\text{TVNIG}(\tenM,\tenA,\bigotimes_{d=1}^D\matdel_d,\kappa)$ to refer to this distribution.
We note that like the TVST, the TVGH, TVVG, and TVNIG are similar in form to their multivariate \citep[][respectively]{browne15,smcnicholas17,karlis09} and their matrix variate \citep{gallaugher19} counterparts.
In fact, as with the tensor variate normal distribution, the four skewed distributions presented here are closely related to their lower order counterparts. These relationships are summarized in the form of the following corollary to Theorem \ref{the:equiv}.
\begin{corollary}
Let $\text{TVD}_{\bf n}(\tenM,\tenA,\bigotimes_{d=1}^D\matdel_d,\btheta)$ represent one of the four skewed tensor distributions of dimension ${\bf n}$, where $\btheta$ represents the additional parameters specific to the distribution. Let $\text{MVD}_{n\times p}(\matm,\vecA,\matSig,\matPsi,\btheta)$ represent the corresponding matrix variate distribution. Finally, let ${\text D}(\vecmu,\vecalp,\matSig,\btheta)$ represent the corresponding multivariate distribution. The following statements are then equivalent.
\begin{enumerate}
\item $\fX\sim \text{TVD}_{\vecn}\left(\tenM,\tenA,\bigotimes_{d=1}^D\matdel_d,\btheta\right)$
\item $\mathcal{X}_{(j)}\sim \text{MVD}_{\frac{n^*}{n_j}\times n_j}\left(\matm_{(j)},\vecA_{(j)},\bigotimes_{d\ne j}\matdel_d, \matdel_j,\btheta\right)$
\item $\vecc(\mathcal{X}_{(j)})\sim \text{D}_{n^*}\left(\vecc({\matm_{(j)}}),\vecc({\vecA_{(j)}}),\matdel_j\otimes\bigotimes_{d\ne j}\matdel_d,\btheta\right)$
\end{enumerate}
\end{corollary}
The proof is an easy application of Theorem \ref{the:equiv} and the form of the variance mean mixture model.

\subsection{Expectations}
The expectations for these four distributions can be easily calculated using the following theorem.
\begin{theorem}
Suppose a random order-$D$ tensor $\fX$ of dimension $\vecn$ can be written in the form $\fX=\tenM+W\tenA+\sqrt{W}\fV$, where $\tenM$ and $\tenA$ are $\vecn$ dimensional order-$D$ tensors, $W\in\mathbb{R}^+$ is a positive random variable, and $\fV\sim\mathcal{N}(\mathbb{O},\bigotimes_{d=1}^D\matdel_d)$. Then, $\mathbb{E}[\fX]=\tenM+\mathbb{E}[W]\tenA$.
\label{the:Exp}
\end{theorem}

The proof of this theorem is a trivial use of iterative expectation, and the tensor variate normality of $\fX$ given $W$. 
Therefore, we have the following expectations:
\begin{align}
\fX\sim\text{TVST}\left(\tenM,\tenA,\bigotimes_{d=1}^D\matdel_d,\nu\right)&\implies\mathbb{E}[\fX]=\tenM+\frac{\nu}{\nu-2}\tenA \hspace{0.1in}(\nu>2)\\
\fX\sim\text{TVGH}\left(\tenM,\tenA,\bigotimes_{d=1}^D\matdel_d,\lambda,\omega\right)&\implies\mathbb{E}[\fX]=\tenM+\frac{K_{\lambda+1}(\omega)}{K_{\lambda}(\omega)}\tenA\\
\fX\sim\text{TVVG}\left(\tenM,\tenA,\bigotimes_{d=1}^D\matdel_d,\gamma\right)&\implies\mathbb{E}[\fX]=\tenM+\tenA\\
\fX\sim\text{TVNIG}\left(\tenM,\tenA,\bigotimes_{d=1}^D\matdel_d,\kappa\right)&\implies\mathbb{E}[\fX]=\tenM+\frac{1}{\kappa}\tenA \label{eq:ev_nig}
\end{align}

\begin{theorem}
	If we define the order-$D$ tensor $\fZ \sim \mathcal{N}\left(\mathbb{O}, \bigotimes_{d=1}^D\ident_{d}\right)$, we can use a tucker product \citep{kolda2009} to define $\fV = \fZ \times \matdel^\half = \fZ \times_1 \matdel_1^\half \times_2 \matdel_2^\half \cdots \times_D \matdel_D^\half$. Let an equivalent mode-1 matricized version of $\fV$ be $\vecV_{(1)} = \matdel_1^\half\vecZ_{(1)}\left(\bigotimes_{d=D}^2\matdel_d^\half\right)\T$ and of $\fX$ be $\vecX_{(1)} = \vecM_{(1)}+W\vecA_{(1)}+\sqrt{W}\vecV_{(1)}$. Then	
	\begin{align} 
	\nonumber \text{Cov}\left[\vecc(\fX)\right] &= \vecc(\vecMone)\vecc(\vecMone)\T + \EV[W]\vecc(\vecMone)\vecc(\vecA_{(1)})\T \\
	\nonumber &\quad + \EV[W]\vecc(\vecA_{(1)})\vecc(\vecMone)\T
 + \EV[W^2]\vecc(\vecA_{(1)})\vecc(\vecA_{(1)})\T  + \EV[W]\left(\bigotimes_{d=D}^1\matdel_d\right) \\
	\nonumber \EV\left[\vecX_{(1)}\vecX_{(1)}\T\right] &= \vecM_{(1)}\vecM_{(1)}\T + \nonumber \EV\left[W\right]\vecM_{(1)}\vecA_{(1)}\T + \EV\left[W\right]\vecA_{(1)}\vecM_{(1)}\T \\
	& \quad + \EV\left[W^2\right]\vecA_{(1)}\vecA_{(1)}\T + \EV[W]\matdel_1\times\prod_{d=2}^D\tr\left(\matdel_d\right)\\
	\nonumber \EV\left[\vecX_{(1)}\T\vecX_{(1)}\right] &= \vecM_{(1)}\T\vecM_{(1)} +  \EV\left[W\right]\vecM_{(1)}\T\vecA_{(1)} + \EV\left[W\right]\vecA_{(1)}\T\vecM_{(1)} \\
	& \quad + \EV\left[W^2\right]\vecA_{(1)}\T\vecA_{(1)} + \EV[W]\left(\bigotimes_{d=D}^2\matdel_d\right) \times  \tr\left(\matdel_1\right). 
	\end{align}
	\label{the:Expec}
\end{theorem}
The proof of this theorem is given in Appendix B. Equivalent expressions can be found for different modes of $\fX$ by using different matricizations.

\subsection{Characteristic Functions}
The calculation of the characteristic functions for these four distributions rely on the following theorem.
\begin{theorem}
Suppose a random order-$D$ tensor $\fX$ of dimension $\vecn$ can be written in the form $\fX=\tenM+W\tenA+\sqrt{W}\fV$, where $\tenM$ and $\tenA$ are $\vecn$ dimensional order-$D$ tensors, $W\in\mathbb{R}^+$ is a positive random variable, and $\fV\sim\mathcal{N}(\mathbb{O},\bigotimes_{d=1}^D\matdel_d)$. Then 
$$
\mathcal{C}_{\fX}(\tenT)=\exp\{i\vecc(\tenT)'\vecc(\tenM)\}\int_{0}^{\infty}\exp\left\{iWa-\frac{1}{2}Wb\right\}h(w)dw,
$$
where $a=\vecc(\tenT)'\vecc(\tenA)$, $b=\vecc(\tenT)'\bigotimes_{d=1}^D\matdel_d\vecc(\tenT)$, and $h(w)$ is the probability density function of $W$.
\label{the:CF}
\end{theorem}
\begin{proof}
For the purposes of this proof, let $\vect=\vecc(\tenT), \vecmu=\vecc(\tenM),$ and $\vecalpha=\vecc(\tenA)$. First note that because of the formulation of $\fX$, we have that $\fX|W\sim\mathcal{N}_{\vecn}(\tenM+W\tenA,W\bigotimes_{d=1}^D\matdel_d)$. Using iterative expectation, along with \eqref{eq:CFN}, we then have
\begin{align*}
\mathcal{C}_{\fX}(\tenT)&=\mathbb{E}[\exp\{i\vect'\vecc(\fX)\}]
=\mathbb{E}[\mathbb{E}[\exp\{i\vect'\vecc(\fX)\}|W]]\\
&=\exp\{i\vect'\vecmu\}\mathbb{E}\left[\exp\left\{i\vect'\vecalpha w-\frac{1}{2}w\vect'\bigotimes_{d=1}^D\matdel_d\vect\right\}\right]\\
&=\exp\{i\vect'\vecmu\}\int_0^{\infty}\exp\left\{iwa-\frac{1}{2}wb\right\}h(w)dw,
\end{align*}
where $a=\vect'\vecalpha$, $b=\vect'\bigotimes_{d=1}^D\matdel_d\vect$, and $h(w)$ is the density function of $W$.
\end{proof}
Fortunately the integral in Theorem \ref{the:CF}, can be found in closed form for each of the four distributions considered herein, and the characteristic functions are displayed below. The full derivation is shown in Appendix C. Note that we use the same notation as in the proof of Theorem \ref{the:CF}, and $a$ and $b$ are as defined in the theorem.

If $\fX$ follows a tensor variate skew-$t$ distribution with $\nu$ degrees of freedom, then from Theorem~\ref{the:CF} the characteristic function is
$$
\mathcal{C}_{\fX}(\tenT)=\exp\{i\vect'\vecmu\}\frac{2\frac{\nu}{2}^{\frac{\nu}{2}}K_{-\frac{\nu}{2}}(\sqrt{\nu b})
}{\Gamma\left(\frac{\nu}{2}\right)\left(\frac{b}{\nu}\right)^{-\frac{\nu}{4}}}\mathcal{C}_{\text{GIG}}\left(a~|~b,\nu,-\frac{\nu}{2}\right).
$$
If $\fX$ follows a tensor variate variance gamma distribution with concentration parameter $\omega$ and index parameter $\lambda$, then from Theorem \ref{the:CF} the characteristic function is
$$
\mathcal{C}_{\fX}(\tenT)=\exp\{i\vect'\vecmu\}\frac{K_{\lambda}(\sqrt{(\omega+b)\omega})}{K_{\lambda}(\omega)\left(\frac{\omega+b}{\omega}\right)^{\frac{\lambda}{2}}}\mathcal{C}_{\text{GIG}}(a~|~\omega+b,\omega,\lambda).
$$
If $\fX$ follows a tensor variate variance gamma distribution with concentration parameter $\gamma$, then from Theorem \ref{the:CF} the characteristic function is
$$
\mathcal{C}_{\fX}(\tenT)=\exp\{i\vect'\vecmu\}\frac{\gamma^{\gamma}}{\left(\gamma+\frac{1}{2}b\right)^{\gamma}}\mathcal{C}_{\text{Gamma}}\left(a~\middle|~\gamma,\gamma+\frac{1}{2}b\right),\\
$$
where 
$$
\mathcal{C}_{\text{Gamma}}\left(a~\middle|~\gamma,\gamma+\frac{1}{2}b\right)=\left(1-\frac{2ia}{2\gamma+b}\right)^{-\gamma}
$$
is the characteristic function of a gamma distribution with parameters $\gamma$ and $\gamma+b/2$ evaluated at~$a$. 

We take this time to note that another skewed tensor distribution is easily obtained from the TVVG distribution. In the multivariate case, if $W\sim \text{Exp}(1)$, where $\text{Exp}(\cdot)$ represents the exponential distribution with rate $\lambda$, then this results in the shifted asymmetric Laplace (SAL) distribution \citep{franczak14}. Therefore, due to the close relationship between the tensor and multivariate distributions, the tensor variate SAL (TVSAL) would naturally arise as a special case of the TVVG with $\gamma=1$. This can also be viewed as the TVVG distribution without the ability to model concentration.

If $\fX$ follows a tensor variate variance gamma distribution with concentration parameter $\kappa$, then from Theorem \ref{the:CF} the characteristic function is
$$
\mathcal{C}_{\fX}(\tenT)=\exp\{i\vect'\vecmu\}\int_0^{\infty}\exp\left\{iwa-\frac{1}{2}wb\right\}\frac{\exp\{\kappa\}w^{-\frac{2}{3}}}{\sqrt{2\pi}}\exp\left\{-\frac{1}{2}\left(\kappa^2w+\frac{1}{w}\right)\right\}dw,\\
$$
where
$$
\mathcal{C}_{\text{IG}}(t~|~1,\gamma)=\exp\left\{\gamma\left(1-\sqrt{1-\frac{2it}{\gamma^2}}\right)\right\}
$$
is the characteristic function of the $\text{IG}(1,\gamma)$ distribution.

\subsection{Benefits Over Vectorization}
Just like in the matrix variate case, the tensor observations can be vectorized and then analyzed as a vector; however, there are a few drawbacks to using this method. The first is that the scale matrices $\matdel_d$ allow for the modelling of element dependencies within each mode of the tensor. 

Secondly, the number of free scale parameters is significantly reduced. If we consider an order-$D$ tensor of dimension ${\bf n}$, then the result is an $n^*$ dimensional vector. If no restraints were placed on the scale matrix when modelling the vectorized version, then there would be $n^*(n^*+1)/2$ free scale parameters that would need to be estimated. There are, of course, constraints that could be placed on the scale matrix such as considering the eigenvalue decomposition, or implementing a factor analysis; however, even these methods would fail to provide adequate results when the dimension surpasses even 100, which is easily obtained with even low dimensions in each mode such as a $5 \times5 \times 5$ order-3 tensor. By modelling with one of the proposed tensor variate distributions, parameter estimation of the scale parameters is restricted to estimating $D$ lower dimensional scale matrices leading to $\sum_{d=1}^Dn_d(n_d+1)/2$ free scale parameters. Therefore, in the previous case of a $5\times 5\times 5$ order-3 tensor, there would be only 45 scale parameters when using a tensor distribution in comparison to 7875 scale parameters in an unconstrained scale matrix when vectorizing.

\section{Parameter Estimation}
Parameter estimation can proceed in one of two ways. Moreover, both of these are based on an expectation conditional maximization \citep[ECM;][]{meng93} algorithm. The first is founded on the flip-flop algorithm based on the algorithm proposed by \cite{manceur13}. Suppose we observe a sample of $N$ tensors $\mathcal{X}=(\tenX_1, \tenX_2, \ldots, \tenX_N)$ from one of the four skewed, $\vecn$ dimensional, order-$D$ tensor variate distributions. We proceed as if the observed data is incomplete, and introduce the latent variables $w_i$. 

The complete log likelihood is then
\begin{align}
\nonumber \ell_{C}(\bvtheta| \mathcal{X},\vecw)&= C-\sum_{d=1}^D\frac{Nn^*}{2n_d}\log(|\matdel_d|)+h(w~|~\btheta)\\ 
\nonumber & +\frac{1}{2}\sum_{i=1}^N \vecc(\tenX_i-\tenM)\T\bigotimes_{d=1}^D\matdel_{d}^{-1}\vecc(\tenA)+\frac{1}{2}\sum_{i=1}^N \vecc(\tenA)\T\bigotimes_{d=1}^D\matdel_{d}^{-1}\vecc(\tenX_i-\tenM)
\\
& -\frac{1}{2}\sum_{i=1}^N\frac{1}{w_i}\left(\vecc(\tenX_i-\tenM)\T\bigotimes_{d=1}^D\matdel_{d}^{-1}\vecc(\tenX_i-\tenM)\right)-\frac{1}{2}\sum_{i=1}^Nw_i\vecc(\tenA)\T\bigotimes_{d=1}^D\matdel_{d}^{-1}\vecc(\tenA) \label{eq:ll},
\end{align}
where $C$ is a constant that does not depend on the parameters, and $h(w~|~\bvtheta)$, is the density of $W$ and is dependent on the tensor variate distribution of interest.

We proceed by using an ECM algorithm described below.

{\bf 1) Initialization}: Initialize the parameters $\tenM,\tenA,\vecDelta_j\text{'s}, \bvtheta$.

{\bf 2) E Step}: Update $a_i, b_i, c_i$, where
\begin{equation*}\begin{split}
a_{i}&=\mathbb{E}(W_{i}~|~\tenX_i,\hat{\bvtheta}),\quad
b_{i}=\mathbb{E}\left(\frac{1}{W_{i}}~\bigg|~\tenX_i,\hat{\bvtheta}\right),\quad
c_{i}=\mathbb{E}(\log W_{i}~|~\tenX_i,\hat{\bvtheta}).\\
\end{split}\end{equation*}
As usual, all expectations are conditional on current parameter estimates; however, to avoid cluttered notation, we do not use iteration-specific notation herein. Although these expectations are dependent on the distribution in question, it can be shown that in each case, the conditional distributions follow a GIG distribution and, therefore, these expectations can be calculated in a mathematically tractable form. The exact distributions and expectations are given in Appendix D. \\
{\bf 3) First CM Step}: Update the parameters $\tenM,\tenA,\bvtheta$. 
\begin{align*}
\hat{\tenM}^{(t+1)}&=\frac{\sum_{i=1}^N\tenX_i\left(\overline{a}^{(t+1)}b^{(t+1)}_i-1\right)}{\sum_{i=1}^N\overline{a}b_i^{(t+1)}-N} \numberthis \label{eq:muup}\\
\hat{\tenA}^{(t+1)}&=\frac{\sum_{i=1}^N\tenX_i\left(\overline{b}-b^{(t+1)}_i\right)}{\sum_{i=1}^Na_i^{(t+1)}\overline{b}-N} \numberthis \label{eq:Aup}
\end{align*}

The updates for the additional parameters, $\bvtheta$, are equivalent to the single component updates for the matrix variate counterparts found in \cite{gallaugher18a}. The exact updates for each distribution are presented in Appendix E.

{\bf 4) Additional CM Steps}: Update $\matdel_j$
\begin{equation}\label{eq:updateDelta}
\begin{split}
\hat{\matdel}_j^{(t+1)}
&=\frac{n_d}{Nn^*}\sum_{i=1}^N b^{(t+1)}_i\left(\vecX{_i}_{(j)}-\hat{\matm}_{(j)}^{(t+1)}\right)\bigotimes_{\substack{d\ne j}}\hat{\matdel}_d^{{-1}}\left(\vecX{_i}_{(j)}-\hat{\matm}_{(j)}^{(t+1)}\right)\T \\ 
&-\hat{\vecA}_{(j)}^{(t+1)}\bigotimes_{\substack{d\ne j}}\hat{\matdel}_d^{{-1}}\left(\vecX{_i}_{(j)}-\hat{\matm}_{(j)}^{(t+1)}\right)\T
-\left(\vecX{_i}_{(j)}-\hat{\matm}_{(j)}^{(t+1)}\right)\bigotimes_{\substack{d\ne j}}\hat{\matdel}_d^{{-1}}(\hat{\vecA}_{(j)}^{(t+1)})\T\\
&+a_i^{(t+1)}\hat{\vecA}_{(j)}^{(t+1)}\bigotimes_{\substack{d\ne j}}\hat{\matdel}_d^{{-1}}(\hat{\vecA}_{(j)}^{(t+1)})\T
\end{split}
\end{equation}

{\bf 5) Check Convergence}: If not converged repeat steps 2--5 until convergence.

It is clear that this proposed algorithm is mathematically tractable; however, the flip-flop method for updating the scale matrices is computationally intensive, and possibly infeasible. The possibility of being infeasible arises because the updates for the scale matrices requires matricization along each of the $D$ modes. Therefore, we propose another method for Step 4 of the algorithm that only requires the matricization of each tensor observation along the first mode.

The quadratic forms in \eqref{eq:ll} can be re-expressed as traces of the mode one tensor matricizations and thus the complete log likelihood can be expressed two ways. The first is
\begin{align}
\nonumber \ell_{C}(\bvtheta| \mathcal{X},\vecw) &= C-\sum_{d=1}^D\frac{Nn^{*}}{2n_d}\log(|\matdel_d|)+h(w~|~\bvtheta)\\ 
\nonumber & +\frac{1}{2}\sum_{i=1}^N\frac{1}{w_i} \sum_{j=1}^{n^{*}_{3:D}} \tr\left[\vecDelta_1\inv \vecX_{(1)ij}\T \vecDelta_2\inv \vecX_{(1)ij}\right]+\frac{1}{2}\sum_{i=1}^N\sum_{j=1}^{n^{*}_{3:D}} \tr\left[\vecDelta_1\inv \vecA_{(1)j}\T \vecDelta_2\inv \vecX_{(1)ij}\right]
\\
& -\frac{1}{2}\sum_{i=1}^N\sum_{j=1}^{n^{*}_{3:D}} \tr\left[\vecDelta_1\inv \vecX_{(1)ij}\T \vecDelta_2\inv \vecA_{(1)j}\right]-\frac{1}{2}\sum_{i=1}^Nw_i \sum_{j=1}^{n^{*}_{3:D}} \tr\left[\vecDelta_1\inv \vecA_{(1)j}\T\vecDelta_2\inv \vecA_{(1)j} \right], \label{eq:ll2}
\end{align}
where $n^{*}_{3:D} = \prod_{d=3}^{D} n_d$, $\vecX_{(1)ij} = (\ident_{n_2} \otimes \vecejT \bigotimes_{d=3}^D \cholDeltaI{d} )\breve{\vecX}_{(1)i}$, $\breve{\vecX}_{(1)i} = \vecX_{(1)i} - \vecMone$, $\vecA_{(1)j} = (\ident_{n_2} \otimes \vecejT \bigotimes_{d=3}^D \cholDeltaI{d} )\vecA_{(1)}$, $\cholDeltaI{d}$ is the Cholesky decomposition of $\vecDelta_d\inv$, and $\vecej$ is a kronecker product of unit basis vectors.

An alternative form for \eqref{eq:ll} can be derived using the tensor commutative operator, defined in \cite{ohlson2013}, to permute the rows of the tensor matricizations and the entries of the kronecker products. It exchanges the second and $l^{th}$ elements in the sequence, where $3 \leq l \leq D$. We denote these modifications by the superscript $l2$. The second expression for the complete log-likelihood is defined thusly
\begin{align}
\nonumber \ell_{C}(\bvtheta| \mathcal{X},\vecw) &= C-\sum_{d=1}^D\frac{Nn^{*}}{2n_d}\log(|\matdel_d|)+h(w~|~\bvtheta)\\ 
\nonumber & +\frac{1}{2}\sum_{i=1}^N\frac{1}{w_i} \sum_{j=1}^{n^{*}_{2:D/l}} \tr\left[\vecDelta_1\inv (\vecX_{(1)ij}^{l2})\T \vecDelta_l\inv \vecX_{(1)ij}^{l2}\right]+\frac{1}{2}\sum_{i=1}^N\sum_{j=1}^{n^{*}_{2:D/l}} \tr\left[\vecDelta_1\inv (\vecA_{(1)j}^{l2})\T \vecDelta_l\inv \vecX_{(1)ij}^{l2}\right]
\\
& -\frac{1}{2}\sum_{i=1}^N\sum_{j=1}^{n^{*}_{2:D/l}} \tr\left[\vecDelta_1\inv (\vecX_{(1)ij}^{l2})\T \vecDelta_l\inv \vecA_{(1)j}^{l2}\right]-\frac{1}{2}\sum_{i=1}^Nw_i \sum_{j=1}^{n^{*}_{2:D/l}} \tr\left[\vecDelta_1\inv (\vecA_{(1)j}^{l2})\T\vecDelta_l\inv \vecA_{(1)j}^{l2} \right], \label{eq:ll3}
\end{align}
where $n^{*}_{2:D/l} = \prod_{\substack{j=2\\j\ne l}}^D n_j$, $\vecX_{(1)ij}^{l2} = (\ident_{n_l} \otimes \vecejT \bigotimes_{\substack{j=2\\j\ne l}}^D \cholDeltaI{d} )\breve{\vecX}_{(1)i}^{l2}$ and $\vecA_{(1)j}^{l2} = (\ident_{n_l} \otimes \vecejT \bigotimes_{\substack{j=2\\j\ne l}}^D \cholDeltaI{d} )\vecA_{(1)}^{l2} $.

The ECM algorithm described above is modified in step four by replacing  \eqref{eq:updateDelta} with the following three equations, one for $\hat{\vecDelta}_1$, $\hat{\vecDelta}_2$ and the general update for $\hat{\vecDelta}_l$:
\begin{align}
\nonumber \hat{\vecDelta}_1^{(t+1)} =& \frac{n_1}{Nn^{*}}\sum_{i=1}^N \sum_{j=1}^{n^{*}_{3:D}} -\hat{\vecA}_{(1)j}\T\hat{\vecDelta}_2\inv\vecX_{(1)ij} - \vecX_{(1)ij}\T\hat{\vecDelta}_2\inv\hat{\vecA}_{(1)j} \\
 &+b^{(t+1)}_i\vecX_{(1)ij}\T\hat{\vecDelta}_2\inv\vecX_{(1)ij} + a^{(t+1)}_i\hat{\vecA}_{(1)j}\T\hat{\vecDelta}_2\inv\hat{\vecA}_{(1)j},\\
\nonumber \hat{\vecDelta}_2^{(t+1)} =& \frac{n_2}{Nn^{*}}\sum_{i=1}^N \sum_{j=1}^{n^{*}_{3:D}} -\vecX_{(1)ij}\hat{\vecDelta}_1\inv\hat{\vecA}_{(1)j}\T - \hat{\vecA}_{(1)j}\hat{\vecDelta}_1\inv\vecX_{(1)ij}\T\\
&+b^{(t+1)}_i\vecX_{(1)ij}\hat{\vecDelta}_1\inv\vecX_{(1)ij}\T + a^{(t+1)}_i\hat{\vecA}_{(1)j}\hat{\vecDelta}_1\inv\hat{\vecA}_{(1)j}\T,\\
\nonumber \hat{\vecDelta}_l^{(t+1)} =& \frac{n_l}{Nn^{*}}\sum_{i=1}^N \sum_{j=1}^{n^{*}_{2:D/l}} -\vecX_{(1)ij}^{l2}\hat{\vecDelta}_1\inv\left(\hat{\vecA}_{(1)j}^{l2}\right)\T - \hat{\vecA}_{(1)j}^{l2}\hat{\vecDelta}_1\inv(\vecX_{(1)ij}^{l2})\T \\
&+b^{(t+1)}_i\vecX_{(1)ij}^{l2}\hat{\vecDelta}_1\inv\left(\vecX_{(1)ij}^{l2}\right)\T + a^{(t+1)}_i\hat{\vecA}_{(1)j}^{l2}\hat{\vecDelta}_1\inv\left(\hat{\vecA}_{(1)j}^{l2}\right)\T.   
\end{align}

The ECM algorithm is implemented in version 1.5.3 of the Julia programming language \citep{bezanson2017}. Further details are available in Appendix F.

\subsection{Identifiability}

As was discussed in \cite{dutilleul99}, \cite{Anderlucci15} and \cite{gallaugher18a} for parameter estimation in the matrix variate case, and in \cite{tait2020} for the order-$D$ case, the estimates of $\vecDelta_d$ are unique only up to a multiplicative constant. Indeed, if we let $d_k = 1/\delta_{k,(1,1)}$, where $\delta_{k,(1,1)}$ is the first entry in $\vecDelta_d$ then
\begin{equation}\label{eq:identif}
\bigotimes_{d=1}^D\vecDelta_d = \frac{1}{\prod_{k=2}^{D}d_k}\vecDelta_1 \otimes \bigotimes_{k=2}^D d_k\vecDelta_k,  
\end{equation}
and therefore the likelihood is unchanged. However, we notice that $\bigotimes_{d=1}^D\vecDelta_d = \bigotimes_{d=1}^D\tilde{\vecDelta}_d$, where $\tilde{\vecDelta}_d$ are the terms on the right-hand side of \eqref{eq:identif} so the estimate of the Kronecker product would be unique. 

There are several options for solving this problem of non identifiability. One possible solution is to set the first diagonal element of $\matdel_d$ for $1\le d\le D-1$, as generalization of the method used by \cite{gallaugher18a}. Another is to set $\tr(\matdel_d)=n_d$ again for $1\le d\le D-1$. This leads to a total of 
$$
\sum_{d=1}^D \frac{n_d(n_d+1)}{2}-D+1
$$
free parameters from the scale matrices.

\section{Examples}

\subsection{Simulation study}
We conduct a simulation study to investigate the effect of different sample and tensor sizes to investigate if we can effectively estimate the model parameters. 
The simulations are conducted using order-3 tensors. We consider sample sizes $N\in\{50, 100, 150\}$. The $n^*$ quantity is used to measure the different dimensions of the tensors. Its values include 512, 729, 1331, 2197, 3375 and 4813. While these values of $n^*$ can equate to any product of dimension lengths, we consider equal dimension lengths of 8, 9, 11, 13, 15 and 17, respectively, for order-3 tensors. For each combination of $N$ and $n^*$, 100 datasets are simulated. We compare the ECM algorithm for the four skewed tensor variate distributions to the flip-flop algorithm for the tensor variate normal distribution described in \cite{manceur13}. 

We use the relative error to determine how close the estimated model parameters are to the true parameters. It is defined as $\frac{\norm{\hat{\vecV} - \vecV}_F}{\norm{\vecV}_F}$, where $\norm{\cdot}_F$ is the Frobenius matrix norm, $\hat{\vecV}$ is the estimated parameter value, and $\vecV$ is the true parameter value used to generate the simulated data. The smaller this ratio is, the less error is present in the model's parameter estimates.

\subsubsection{Normal Data}
We first consider how these proposed distributions perform when fitted to TVN data.
The ECM and flip-flop algorithms all converge in three iterations. Figure \ref{fig:sim_norm_m} visualizes the mean and 95\% confidence intervals for the relative error in $\mathbb{E}[\fX]$ across the values of $N$ and $n^*$. As expected, the TVN estimates $\tenM$ well. The TVGH and TVNIG have nearly identical performance, which does not degrade as $N$ and $n^*$ increase in size. Moreover, the performance is fairly similar to the TVN. The other three tensor variate models do a poor job of estimating $\mathbb{E}[\fX]$, but the performance improves as the sample size increases. Moreover, the performance is better for lower values of $n^*$.  
\begin{figure}[!htb]
	\begin{center}
		\includegraphics[height=0.5\textwidth]{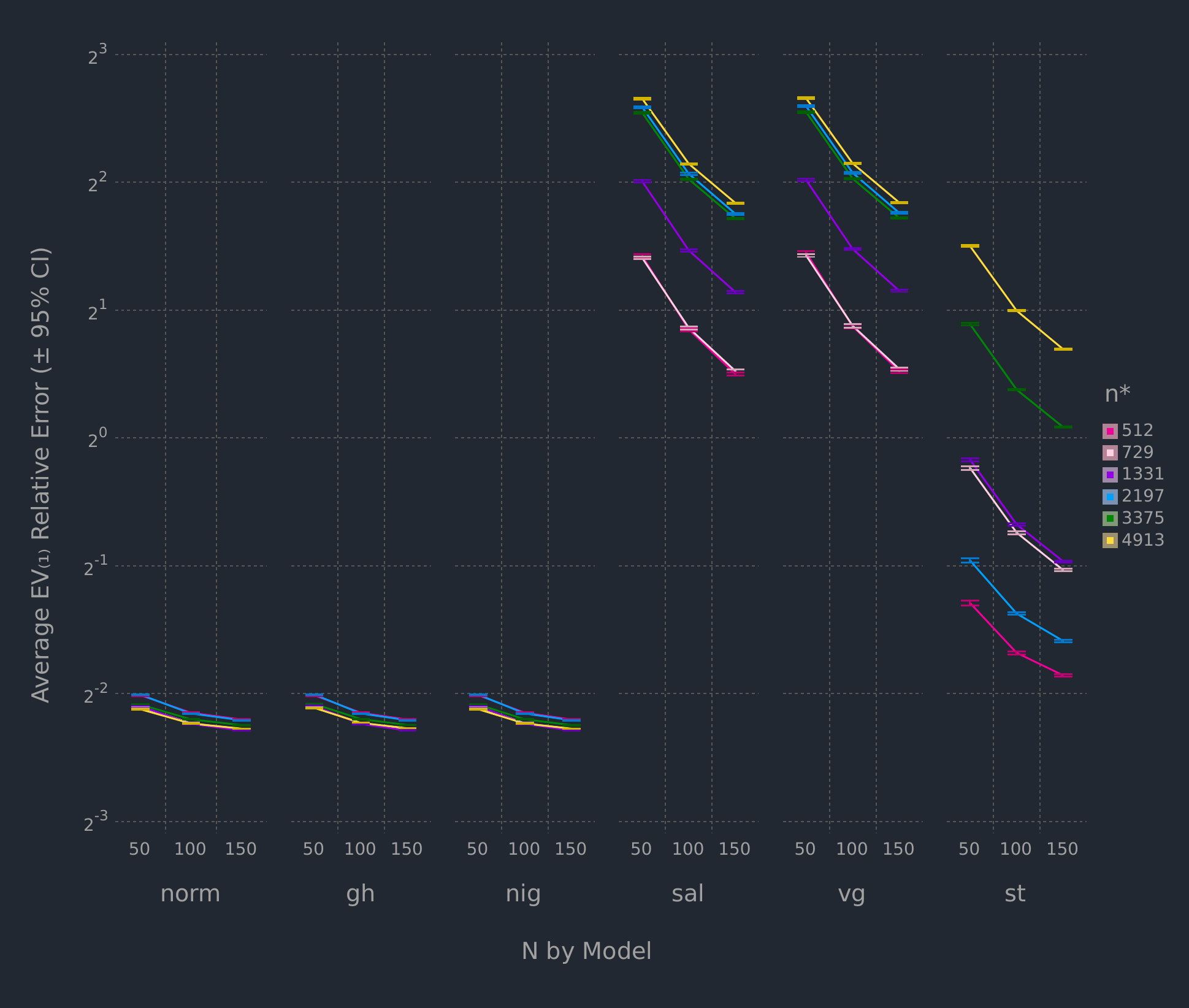}
		\caption{Average and 95\% confidence intervals for the relative error in the mode-1 matricization of $\mathbb{E}[\fX]$ for the simulation on the normal data.}
		\label{fig:sim_norm_m}
	\end{center}
\end{figure}  

A different picture emerges when we look at the relative error in $\bigotimes_{d=1}^D\vecDelta_d$, visualized in Figure~\ref{fig:sim_norm_kp}. 
The TVN performs the worst of all the distributions considered. Moreover, the performance of the skewed distributions does not seem to be greatly affected by the dimension. Finally, the performance does not appear to change for the different sample sizes considered. Further details and results can be found in Appendix G.
\begin{figure}[!htb]
	\begin{center}
		\includegraphics[height=0.5\textwidth]{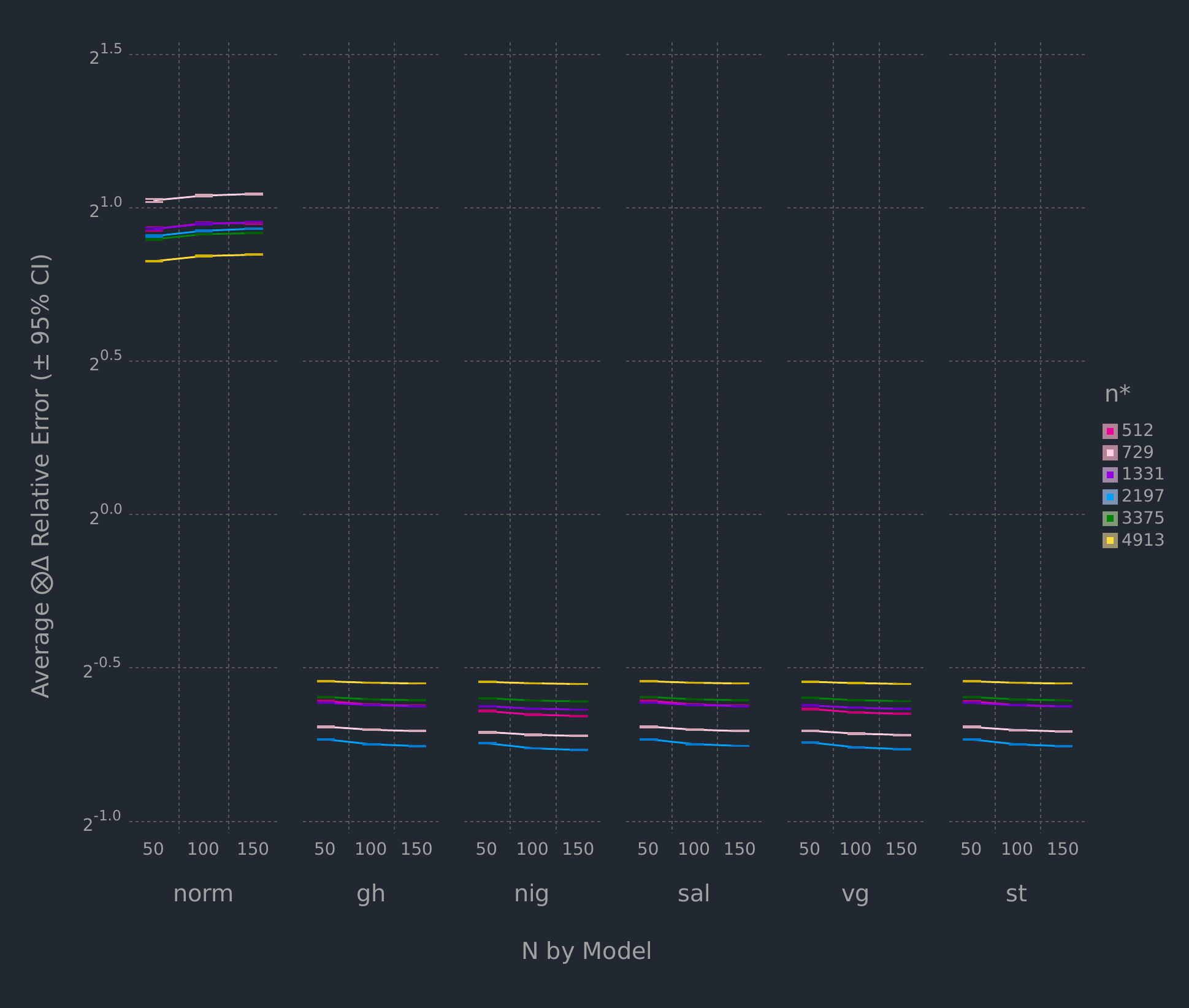}
		\caption{Average and 95\% confidence intervals for the relative error in $\bigotimes_{d=1}^D\vecDelta_d$ for the simulation on the normal data.}
		\label{fig:sim_norm_kp}
	\end{center}
\end{figure}  

\subsubsection{Skewed Data}
We now consider simulations involving skewed data. We used \eqref{eq:VMM} to generate data from a TVST distribution with $\nu = 4$. The different models had a lot of variation in the number of iterations they took to converge to a solution. Typically the normal distribution converges in a median of 3 iterations and the skewed distributions converge in a median of 4--6 iterations.
\begin{figure}[!htb]
	\begin{center}
		\includegraphics[height=0.5\textwidth]{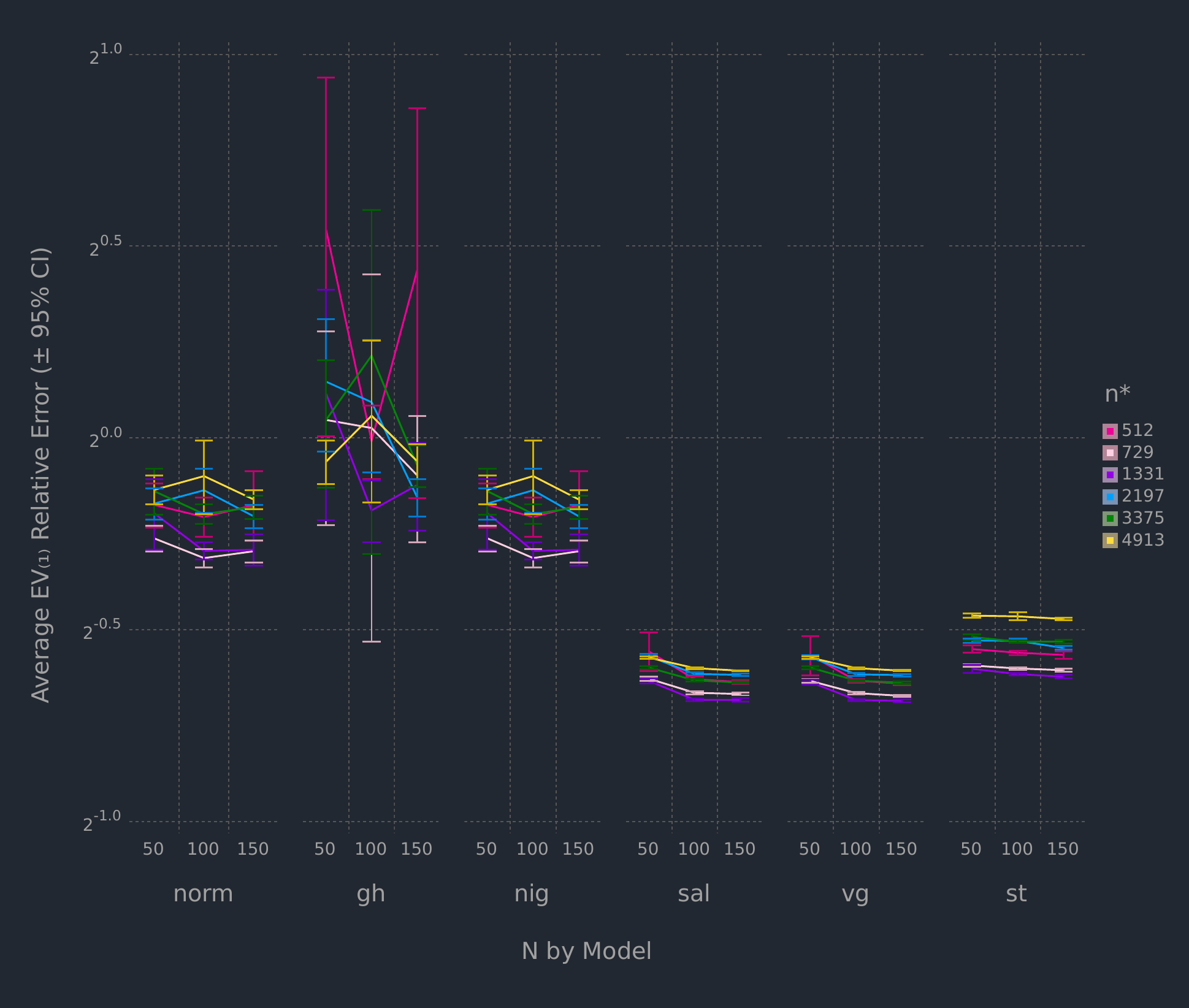}
		\caption{Average and 95\% confidence intervals for the relative error in the mode 1 matricization of $\mathbb{E}[\fX]$ for the simulation of the skewed data.}
		\label{fig:sim_skew_m}
	\end{center}
\end{figure}  

Figure \ref{fig:sim_skew_m} visualizes the mean and 95\% confidence intervals for the relative error in $\mathbb{E}[\fX]$ for the 100 datasets across the values of $N$ and $n^*$.  The TVST, TVVG and TVSAL models perform quite well, accurately estimating $\mathbb{E}[\fX]$ in all scenarios. The TVN and TVNIG distributions perform a little worse, but the estimates are still fairly good. When looking at the TVGH results, however, they are highly variable. This can potentially be explained by the estimated values of $\omega$ and $\lambda$. The resulting GIG distribution (the distribution of the latent variables $W_i$) highly deviate from the $ \text{Inv-Gamma}\left(\frac{\nu}{2},\frac{\nu}{2}\right)$ distribution used to generate the data.  

\begin{figure}[!htb]
	\begin{center}
		\includegraphics[height=0.5\textwidth]{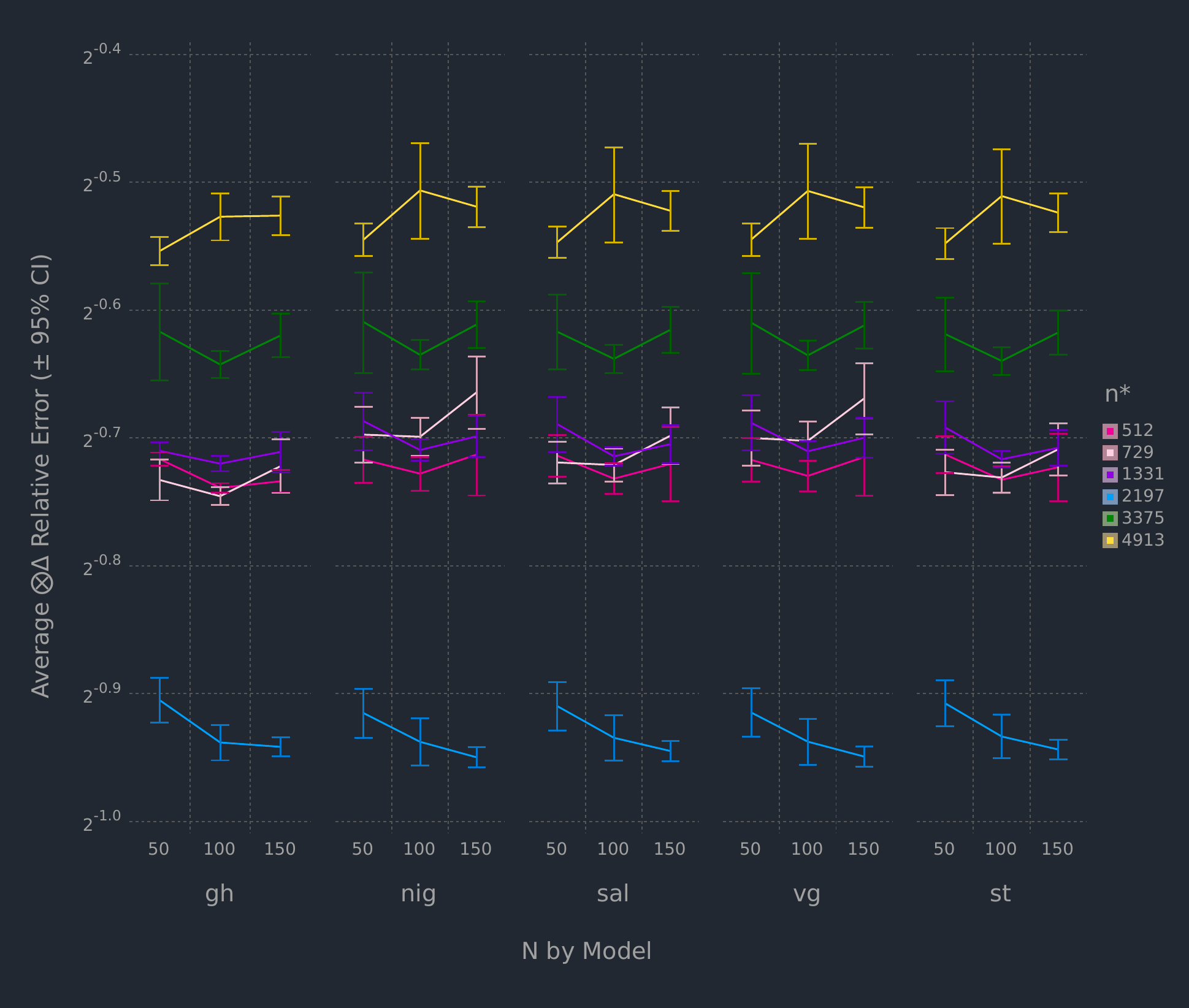}
		\caption{Average and 95\% confidence intervals for the relative error in $\bigotimes_{d=1}^D\vecDelta_d$ for the simulation on the skewed data.}
		\label{fig:sim_skew_kp}
	\end{center}
\end{figure}  

Figure \ref{fig:sim_skew_kp} indicates all the skewed tensor variate models do an excellent job of estimating $\bigotimes_{d=1}^D\vecDelta_d$ across the range of $N$ and $n^*$ values. Their performance degrades only slightly for the two largest values of $n^*$. The TVN has a median relative error of nearly 32. Further details and results for this simulation are presented in Appendix G.

\subsection{Image analysis}

We now consider an analysis of red-green-blue (RGB) images. These images come in the form three colour intensity matrices (red, green, and blue) ``stacked" on top of each other, thus creating an order-3 tensor.
The images come from the CIFAR-100 data set \citep{cifar100}. We choose images of maple trees that had green or yellow leaves and came from the following CIFAR-100 class hierarchy: superclass trees $\rightarrow$ class maple. These tensors had an $n^* = 3072$, making them comparable to the $n^*=3375$ results in our simulation. 
Figure \ref{fig:img_raw} is an example of one of the images in our sample of 207 tensors. 
\begin{figure}[!htb]
	\begin{center}
		\includegraphics[height=0.25\textwidth]{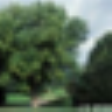}
		\caption{An image of a maple tree from the CIFAR-100 data set.}
		\label{fig:img_raw}
	\end{center}
\end{figure}  

The BIC is used to select the best model for this data, and the skewed models all outperform the normal model (Figure~\ref{fig:img_bic}). The TVNIG model obtains the best performance, with a $\text{BIC} = 2.978\times 10^6$.

\begin{figure}[!htb]
	\begin{center}
		\includegraphics[height=0.4\textwidth]{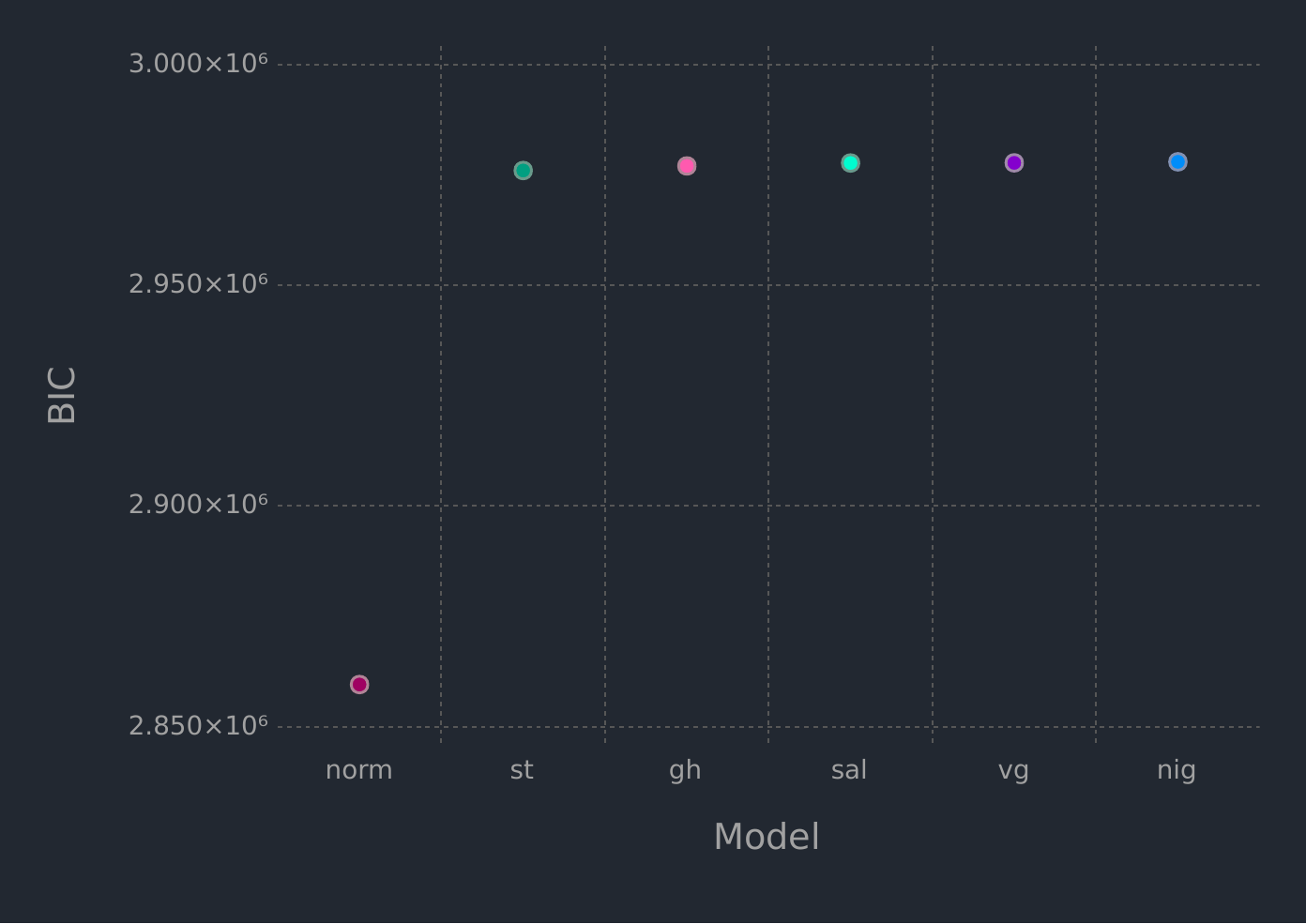}
		\caption{BIC results from the image analysis.}
		\label{fig:img_bic}
	\end{center}
\end{figure}  

The model parameters of the TVNIG model are visualized in Figures \ref{fig:img_mean} to \ref{fig:img_scale}. Figure \ref{fig:img_mean} shows the coloured image that results from the estimated  $\mathbb{E}[\fX]$ tensor. The sky, tree trunk and branches are clearly visible. In Figure \ref{fig:img_mean_slice} we visualize each slice of the estimated $\mathbb{E}[\fX]$ tensor. The sky and tree trunk remain clearly distinguishable for each slice. It is interesting to note that the intensity for the leaves of the tree appear to be slightly higher in the green and red slices compared to the blue slice of the tensor. Moreover, the sky portion of the images show slightly higher intensities in the blue slice.
\begin{figure}[!htb]
	\begin{center}
		\includegraphics[height=0.4\textwidth]{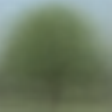}
		\caption{An image of the $\mathbb{E}[\fX]$ tensor from the NIG model.} 
		\label{fig:img_mean}
	\end{center}
\end{figure}  
\begin{figure}[!htb]
	\begin{center}
		\includegraphics[height=0.4\textwidth]{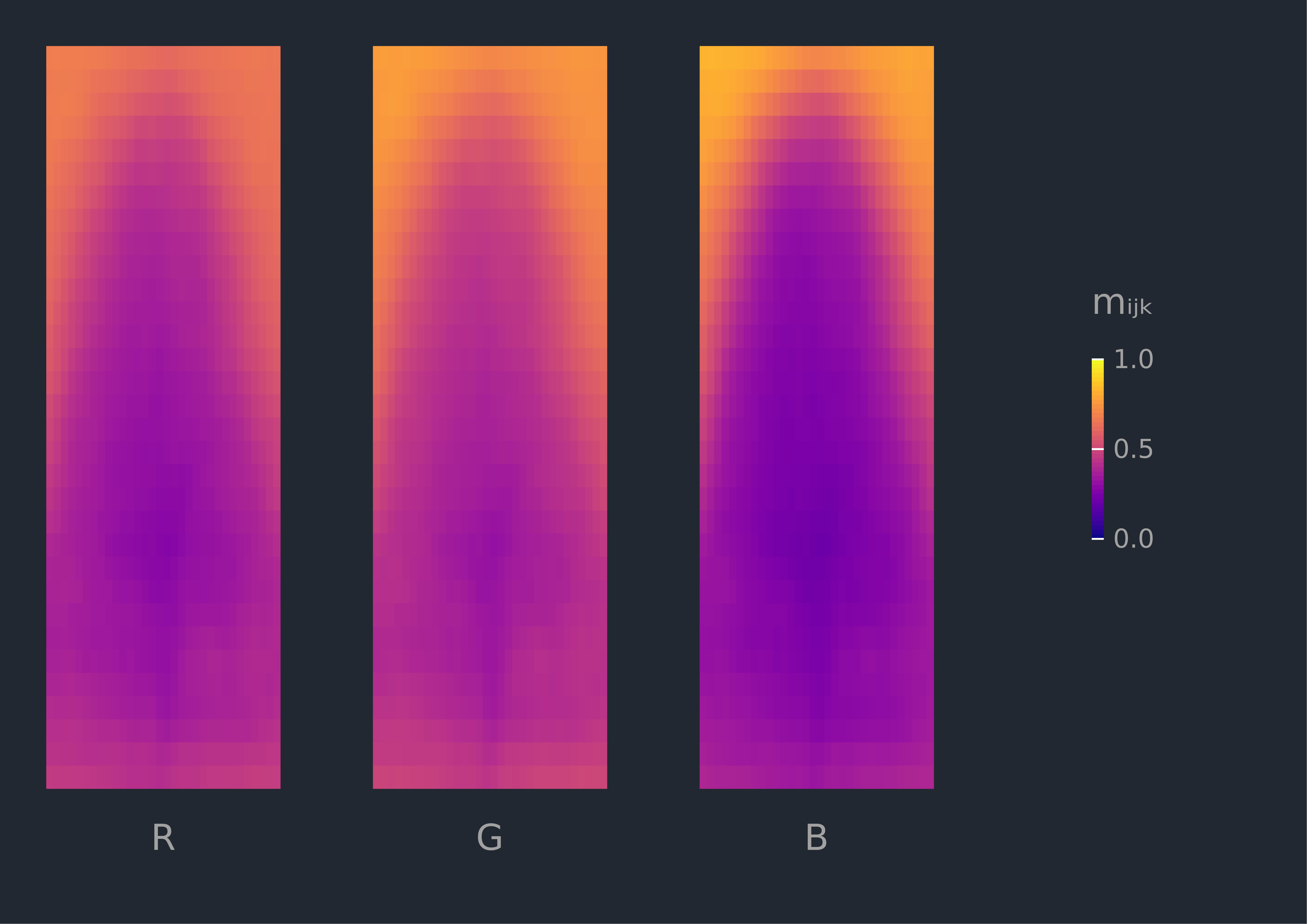}
		\caption{Slices of the $\mathbb{E}[\fX]$ tensor from the NIG model.}
		\label{fig:img_mean_slice}
	\end{center}
\end{figure}  

The estimated skewness tensor slices of $\tenA$ are visualized in Figure \ref{fig:img_skew_slice}. It is clear that each colour slice has different skewness patterns. The sky tends to have the lowest skewness, a pattern accentuated in the ``B" slice. The ``R" slice has the most positive skewness, concentrated in the trunk and body of the trees. 
\begin{figure}[!htb]
	\begin{center}
		\includegraphics[height=0.4\textwidth]{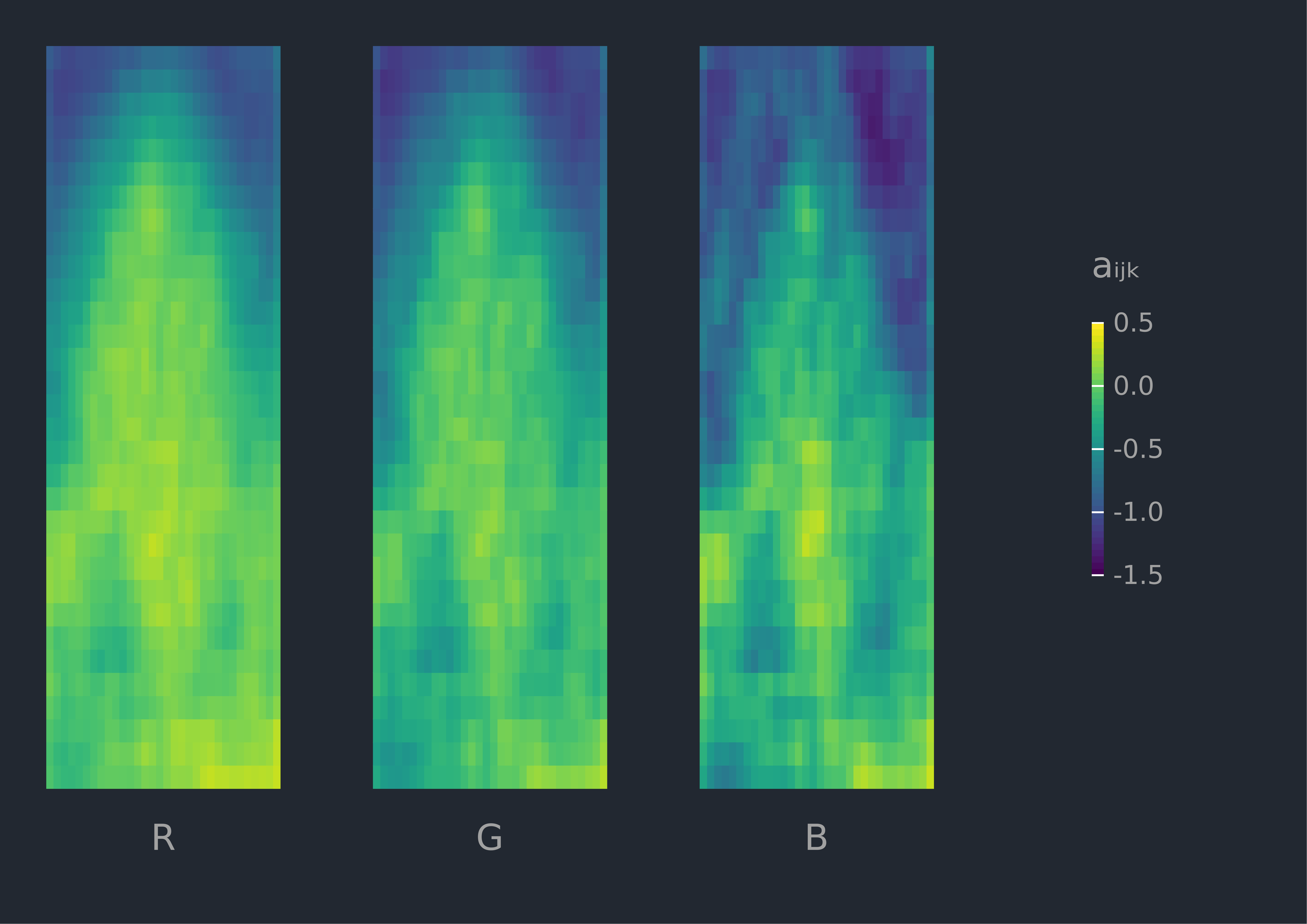}
		\caption{Slices of the $\tenA$ tensor from the NIG model.}
		\label{fig:img_skew_slice}
	\end{center}
\end{figure} 

The estimated variability in each of the three modes is visualized in Figure~\ref{fig:img_scale}, where each scale matrix $\matdel_d$ is visualized as a heatmap. 
\begin{figure}[!htb]
	\begin{center}
		\includegraphics[height=0.4\textwidth]{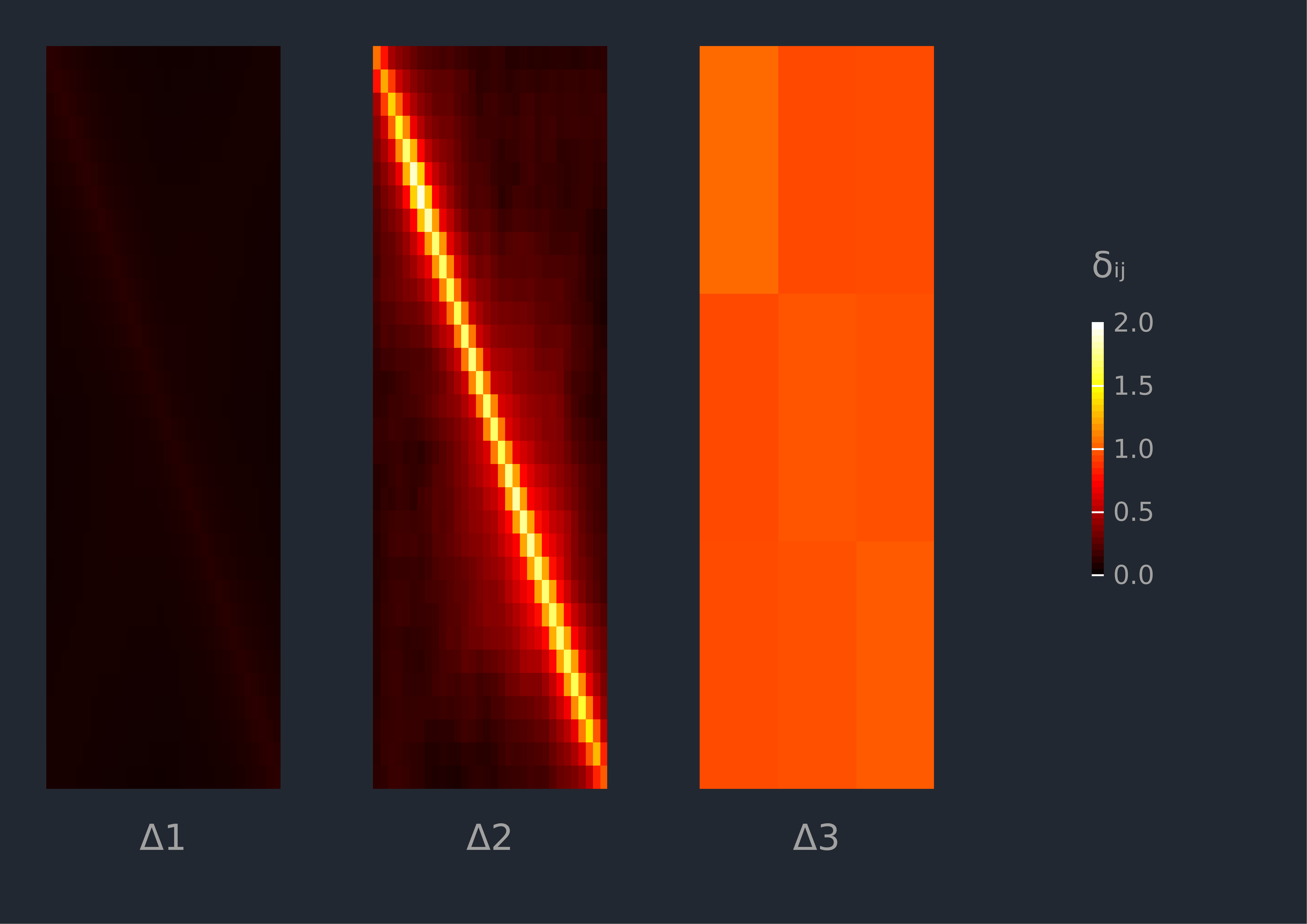}
		\caption{Scale matrices, $\Delta_d$, from the NIG model}
		\label{fig:img_scale}
	\end{center}
\end{figure} 

The rows ($\matdel_1$) have little variation. The columns ($\matdel_2$) exhibit a pattern of covariation consistent with images, one that decreases as the distance between pixels increases. All three slices ($\matdel_3$) have moderate level of variation. 
The location tensor $\tenM$ and the correlation matrices derived from the $\matdel_d$'s are visualized in Appendix H.

\section{Discussion}
In this paper, we derived a total of four skewed tensor variate distributions from a tensor variate normal variance mean mixture model. The densities, as well as expectations, characteristic functions, and two methods for parameter estimation were discussed. In addition parameter estimation was considered using an ECM algorithm. The distributions considered herein can be viewed as extensions of their multivariate and matrix variate counterparts. 

All of these models were considered in two simulation studies. The first considered tensor variate normal data, and the second considered skewed data. In these simulations, the skewed models were compared with the tensor variate normal distribution.
The four skewed distributions were fitted to a dataset consisting of coloured images of maple trees. The best model, as determined by the BIC, was the TVNIG model and all of the skewed models resulted in a better BIC than the TVN model. The resulting image of the estimated mean captured the trunk, branches, leaves and sky, This was also seen in the individual modes of the estimated mean tensor of the TVNIG model. 

Other skewed tensor distributions may be easily derived from their multivariate and matrix variate counterparts using hidden truncation methods; however, parameter estimation may become computationally infeasible due to the overall dimensionality of the tensor. This would be of particular concern when incorporating these distributions into the finite mixture model for use in clustering and classification, which will be one topic of future work. Another topic would be dimension reduction techniques and parsimonious models.

\bibliography{tv4}
\appendix
\section{Derivation details for the tensor variate skew $t$ distribution}\label{app_deriv}

Suppose $\fX$ is a random order-$D$ tensor with a $TVST_{{\bf n}}(\tenM,\tenA,\bigotimes_{d=1}^D\matdel_d,\nu)$ distribution. $\fX$ can be written as $$\fX=\tenM+W\tenA+\sqrt{W}\fV$$
where $\tenM$ and $\tenA$ are ${\bf n}$ dimensional tensors, $\fV \sim \mathcal{N}_{{\bf p}}\left(\mathbb{O}, \bigotimes_{d=1}^D\matdel_d\right)$ and $W\sim \text{Inv-Gamma}\left(\frac{\nu}{2},\frac{\nu}{2}\right)$

The inverse Gamma density has the form:
$$ f(w|a,b) = \frac{b^a}{\Gamma(a)}w^{-(a+1)}\exp\left[-\frac{b}{w}\right] $$

It then follows that
$$
\fX|W=w\sim \mathcal{N}_{\bf p}\left(\tenM+w\tenA,w\bigotimes_{d=1}^D\matdel_d\right)
$$
\subsection{Joint density}

The joint density of $\fX$ and $W$ is
\begin{align*}
f(\tenX,w|\bvtheta)&=f(\tenX|W=w)f(w)\\
&= (2\pi)^{\frac{-n^{*}}{2}}\prod_{d=1}^D \left[w^{n_d}|\matdel_d|\right]^{-\frac{n^{*}}{2n_d}} \times\exp\left\{-\frac{1}{2}\vecc(\tenX-\tenM-w\tenA)\T\frac{1}{w}\bigotimes_{d=1}^D\matdel_d^{-1}\vecc(\tenX-\tenM-w\tenA)\right\} \\ 
& \hspace{0.2in} \times \frac{\frac{\nu}{2}^{\frac{\nu}{2}}}{\Gamma(\frac{\nu}{2})}\left[\frac{1}{w}\right]^{\frac{\nu}{2}+1}\exp\left\{\frac{\frac{\nu}{2}}{w}\right\} \\
&= \frac{\frac{\nu}{2}^{\frac{\nu}{2}}}{(2\pi)^{\frac{n^*}{2}}\prod_{d=1}^D| \matdel_d |^{\frac{n^*}{2n_d}}\Gamma(\frac{\nu}{2})}\cdot w^{-\left(\frac{\nu+n^*}{2}+1\right)}\\
& \hspace{0.2in} \times \exp\left\{-\frac{1}{2w}\left(\vecc(\tenX-\tenM-w\tenA)\T\bigotimes_{d=1}^D\matdel_{d}^{-1}\vecc(\tenX-\tenM-w\tenA)+\nu\right)\right\} \numberthis \label{eq:joint}
\end{align*}

This last expression is the same as equation \ref{eqn:joint}. The intermediate steps pertaining to the determinant of the scale matrices are :
\begin{align*}
\prod_{d=1}^D|w\matdel_d|^{-\frac{n^{*}}{2n_d}} &=
\prod_{d=1}^D \left[w^{n_d}|\matdel_d|\right]^{-\frac{n^{*}}{2n_d}}= w^{-\frac{n^{*}}{2}}\prod_{d=1}^D |\matdel_d|^{-\frac{n^{*}}{2n_d}}.
\end{align*}

Using the following identities:
\begin{itemize}
	\item $\vecc(\vecA+\vecE) = \vecc(\vecA) + \vecc(\vecE)$
	\item $\vecc(\alpha\vecA) = \alpha\vecc(\vecA)$
\end{itemize}
the exponential term in \eqref{eq:joint} can be written as:
\begin{align*}
&\exp\left\{-\frac{1}{2w}\left(\vecc(\tenX-\tenM-w\tenA)\T\bigotimes_{d=1}^D\matdel_{d}^{-1}\vecc(\tenX-\tenM-w\tenA)+\nu\right)\right\}\\
&=\exp\left\{-\frac{1}{2w}\left(\left[\vecc(\tenX-\tenM)\T - w\vecc(\tenA)\T\right] \bigotimes_{d=1}^D\matdel_{d}^{-1} \left[\vecc(\tenX-\tenM) - w\vecc(\tenA)\right] +\nu \right)\right\}\\
&=\exp\left\{-\frac{1}{2w}\left(\vecc(\tenX-\tenM)\T\bigotimes_{d=1}^D\matdel_{d}^{-1}\vecc(\tenX-\tenM) -w\vecc(\tenX-\tenM)\T\bigotimes_{d=1}^D\matdel_{d}^{-1}\vecc(\tenA) \right.\right.  \\&\quad
\left. \left. -w\vecc(\tenA)\T\bigotimes_{d=1}^D\matdel_{d}^{-1}\vecc(\tenX-\tenM) + w^2\vecc(\tenA)\T\bigotimes_{d=1}^D\matdel_{d}^{-1}\vecc(\tenA) + \nu \right)\right\}& \\
&=\exp\left\{-\frac{1}{2w}\left(\vecc(\tenX-\tenM)\T\bigotimes_{d=1}^D\matdel_{d}^{-1}\vecc(\tenX-\tenM) -2w\vecc(\tenX-\tenM)\T\bigotimes_{d=1}^D\matdel_{d}^{-1}\vecc(\tenA) \right.\right.\\&\quad
\left. \left. + w^2\vecc(\tenA)\T\bigotimes_{d=1}^D\matdel_{d}^{-1}\vecc(\tenA)  + \nu \right)\right\}
\end{align*}
Grouping the terms that incorporate $w$, facilitates the integration in section \ref{sec:md}.

\begin{align*}
&\exp\left\{ \vecc(\tenX-\tenM)\T\bigotimes_{d=1}^D\matdel_{d}^{-1}\vecc(\tenA)\right\} \times \\
%
&  \hspace{0.2in} \exp\left\{ -\frac{1}{2}\left[ \frac{\vecc(\tenX-\tenM)\T\bigotimes_{d=1}^D\matdel_{d}^{-1}\vecc(\tenX-\tenM) + \nu}{w} + w \vecc(\tenA)\T\bigotimes_{d=1}^D\matdel_{d}^{-1}\vecc(\tenA)\right] \right\}\\
\end{align*}
If we define $$\begin{array}{ccc}
\delta(\cdot)=\vecc(\tenX-\tenM)\T\bigotimes_{d=1}^D\matdel_{d}^{-1}\vecc(\tenX-\tenM) && \rho(\cdot)=\vecc(\tenA)\T\bigotimes_{d=1}^D\matdel_{d}^{-1}\vecc(\tenA),
\end{array}$$
the following expression is obtained:
\begin{equation}\label{eq:exp1}
\exp\left\{\left(\vecc(\tenX-\tenM)\T\bigotimes_{d=1}^D\matdel_{d}^{-1}\vecc(\tenA)\right)\right\} \times \exp\left\{-\frac{1}{2}\left[\frac{\delta(\tenX;\tenM,\bigotimes_{d=1}^D\matdel_{d}^{-1})+\nu}{w}+w\rho(\tenA;\bigotimes_{d=1}^D\matdel_{d}^{-1})\right]\right\}
\end{equation}

\subsection{Marginal Density}\label{sec:md}

Building on equations \ref{eq:joint} and \ref{eq:exp1}, the marginal density of $\tenX$ is
\begin{align}
\nonumber
f(\tenX)&=\int_{0}^{\infty}f(\tenX,w) dw\\
\nonumber
&=\frac{\frac{\nu}{2}^{\frac{\nu}{2}}}{(2\pi)^{\frac{n^*}{2}}\prod_{d=1}^D| \matdel_d |^{\frac{n^*}{2n_d}} \Gamma(\frac{\nu}{2})}\cdot\exp\left\{\vecc(\tenX-\tenM)\T\bigotimes_{d=1}^D\matdel_{d}^{-1}\vecc(\tenA)\right\} \\
\label{eq:md1}
& \hspace{0.25in} \times \int_{0}^\infty w^{-\left(\frac{\nu+n^*}{2}+1\right)}\exp\left\{-\frac{1}{2}\left[\frac{\delta(\tenX;\tenM,\bigotimes_{d=1}^D\matdel_{d}^{-1})+\nu}{w}+w\rho(\tenA,\bigotimes_{d=1}^D\matdel_{d}^{-1})\right]\right\}dw
\end{align}

Using the following change of variables, we can rearrange the integral in equation \ref{eq:md1}:
$$
\begin{array}{ccc}
\text{u} && \text{du} \\
\frac{\sqrt{\rho(\tenA,\bigotimes_{d=1}^D\matdel_{d}^{-1})}}{\sqrt{\delta(\tenX;\tenM,\bigotimes_{d=1}^D\matdel_{d}^{-1})+\nu}}w &&
\frac{\sqrt{\rho(\tenA,\bigotimes_{d=1}^D\matdel_{d}^{-1})}}{\sqrt{\delta(\tenX;\tenM,\bigotimes_{d=1}^D\matdel_{d}^{-1})+\nu}}dw 
\end{array}
$$

$$
\begin{array}{lr}
\text{Integral term 1} & \text{Integral term 2} \\
w^{-\left(\frac{\nu+n^*}{2}+1\right)} & \exp\left\{-\frac{1}{2}\left[\frac{\delta(\cdot)+\nu}{w}+w\rho(\cdot)\right]\right\}dw\\
\left[\frac{\sqrt{\delta(\cdot)+\nu}}{\sqrt{\rho(\cdot)}}u \right]^{-\left(\frac{\nu+n^*}{2}+1\right)} & \exp\left\{-\frac{1}{2}\left[ \frac{\sqrt{\delta(\cdot)+\nu}\sqrt{\delta(\cdot)+\nu}\sqrt{\rho(\cdot)}}{w\sqrt{\rho(\cdot)}} + \frac{w\sqrt{\rho(\cdot)}\sqrt{\rho(\cdot)}\sqrt{\delta(\cdot)+\nu}}{\sqrt{\delta(\cdot)+\nu}} \right] \right\}dw\\
\left[\left[\frac{\delta(\cdot)+\nu}{\rho(\cdot)}\right]^\frac{1}{2}\right]^{-\left(\frac{\nu+n^*}{2}+1\right)} u^{-\left(\frac{\nu+n^*}{2}+1\right)} & \exp\left\{-\frac{1}{2}\left[ \frac{\sqrt{\delta(\cdot)+\nu}\sqrt{\rho(\cdot)}}{u} + \sqrt{\rho(\cdot)}\sqrt{\delta(\cdot)+\nu}u \right] \right\}dw\\
\left[\frac{\delta(\cdot)+\nu}{\rho(\cdot)}\right]^{-\left(\frac{\nu+n^*}{4}+\frac{1}{2}\right)} u^{-\left(\frac{\nu+n^*}{2}+1\right)} & 
\exp\left\{-\frac{\sqrt{\delta(\cdot)+\nu}\sqrt{\rho(\cdot)}}{2}\left[ u + \frac{1}{u} \right] \right\} \left[\frac{\delta(\cdot)+\nu}{\rho(\cdot)}\right]^\frac{1}{2}du\\
\end{array}
$$

Putting the two terms together, we have:

\begin{align*}
\int_{0}^\infty \left[\frac{\delta(\cdot)+\nu}{\rho(\cdot)}\right]^{-\left(\frac{\nu+n^*}{4}+\frac{1}{2}\right)} u^{-\left(\frac{\nu+n^*}{2}+1\right)}\left[\frac{\delta(\cdot)+\nu}{\rho(\cdot)}\right]^\frac{1}{2} \exp\left\{-\frac{\sqrt{\delta(\cdot)+\nu}\sqrt{\rho(\cdot)}}{2}\left[ u + \frac{1}{u} \right] \right\} du & \rightarrow \\
\left[\frac{\delta(\cdot)+\nu}{\rho(\cdot)}\right]^{-\frac{\nu+n^*}{4}} \int_{0}^\infty u^{-\left(\frac{\nu+n^*}{2}+1\right)} \exp\left\{-\frac{\sqrt{\delta(\cdot)+\nu}\sqrt{\rho(\cdot)}}{2}\left[ u + \frac{1}{u} \right] \right\} du
\end{align*}

The integral is now a bessel function of the second kind.

We can write the marginal density, $f(\tenX)$ as:

\begin{align}
\nonumber f_{\text{TVST}}(\tenX|\bvtheta)=&\frac{2\left(\frac{\nu}{2}\right)^{\frac{\nu}{2}}\exp\left\{\vecc(\tenX-\tenM)\T\bigotimes_{d=1}^D\matdel_{d}^{-1}\vecc(\tenA) \right\} }
{(2\pi)^{\frac{n^*}{2}}\prod_{d=1}^D| \matdel_d |^{\frac{n^*}{2n_d}} \Gamma(\frac{\nu}{2})}
\left(\frac{\delta(\tenX;\tenM,\bigotimes_{d=1}^D\matdel_{d}^{-1})+\nu}{\rho(\vecA,\bigotimes_{d=1}^D\matdel_{d}^{-1})}\right)^{-\frac{\nu+n^*}{4}} \\
\label{eq:pdfST}& \times K_{-\frac{\nu+n^*}{2}}\left(\sqrt{\left[\rho\left(\tenA,\bigotimes_{d=1}^D\matdel_{d}^{-1}\right)\right]\left[\delta\left(X;\tenM,\bigotimes_{d=1}^D\matdel_{d}^{-1}\right)+\nu\right]}\right)
\end{align}

\section{Expectations}\label{sec:app_exp}
Here, we provide a proof of Theorem \ref{the:Expec}.
Define the following terms; $\vecn_{2:D} = \prod_{d=2}^Dn_d$, $\breve{\matdel} = \bigotimes_{d=D}^2\matdel_d$ and $\breve{\matdel}^\half = \bigotimes_{d=D}^2\matdel_d^\half$. Then the expectations can be calculated as follows.

\begin{align*}
\EV\left[\vecc(\fX)\vecc(\fX)\T\right] &= \EV\left[\vecc(\vecX_{(1)})\vecc(\vecX_{(1)})\T\right] \\
&= \EV\left[\left\{\vecc(\vecMone)+W\vecc(\vecA_{(1)})+\sqrt{W}(\breve{\matdel}^\half \otimes \vecDelta_1^\half)\vecc(\vecZ_{(1)})\right\} \times \right.\\
& \qquad \left. \left\{\vecc(\vecMone)\T + W\vecc(\vecA_{(1)})\T + \sqrt{W}\vecc(\vecZ_{(1)})\T(\breve{\matdel}^\half \otimes \vecDelta_1^\half)\T\right\}\right]\\
&= \vecc(\vecMone)\vecc(\vecMone)\T + \EV[W]\vecc(\vecMone)\vecc(\vecA_{(1)})\T + \EV[W]\vecc(\vecA_{(1)})\vecc(\vecMone)\T \\
& \qquad + \EV[W^2]\vecc(\vecA_{(1)})\vecc(\vecA_{(1)})\T \\
& \qquad + \EV[W](\breve{\matdel}^\half \otimes \vecDelta_1^\half)\EV\left[\vecc(\vecZ_{(1)})\vecc(\vecZ_{(1)})\T\right](\breve{\matdel}^\half \otimes \vecDelta_1^\half)\T \\
&= \vecc(\vecMone)\vecc(\vecMone)\T + \EV[W]\vecc(\vecMone)\vecc(\vecA_{(1)})\T + \EV[W]\vecc(\vecA_{(1)})\vecc(\vecMone)\T \\
& \qquad + \EV[W^2]\vecc(\vecA_{(1)})\vecc(\vecA_{(1)})\T  + \EV[W]\left(\bigotimes_{d=D}^1\matdel_d\right).
\end{align*}

\begin{align*}
\EV\left[\vecX_{(1)}\vecX_{(1)}\T\right] &= \EV\left[\left\{\vecM_{(1)} + W\vecA_{(1)} + \sqrt{W}\matdel_1^\half\vecZ_{(1)}\breve{\matdel}^\half\right\} \times \left\{\vecM_{(1)} + W\vecA_{(1)} + \sqrt{W}\matdel_1^\half\vecZ_{(1)}\breve{\matdel}^\half\right\}\T \right]\\
&= \vecM_{(1)}\vecM_{(1)}\T + \EV\left[W\right]\vecM_{(1)}\vecA_{(1)}\T + \EV\left[W\right]\vecA_{(1)}\vecM_{(1)}\T + \EV\left[W^2\right]\vecA_{(1)}\vecA_{(1)}\T\\
& \qquad + \EV[W]\vecDelta_1^\half \EV\left[\vecZ_{(1)}\breve{\matdel}\vecZ_{(1)}\T\right] \vecDelta_1^{\frac{\top}{2}}
\end{align*}

The $\vecZ_{(1)}$ term can be broken into row vectors, $\vecz_{(1)} \in \mathbb{R}^{1\times \vecn_{2:D}}$, and therefore
\begin{align*}
\EV\left[\vecz_{(1)}\breve{\matdel}\vecz_{(1)}\T\right] &= \EV\left[\tr\left(\vecz_{(1)}\T\vecz_{(1)}\breve{\matdel}\right)\right] \\
&= \tr\left(\EV\left[\vecz_{(1)}\T\vecz_{(1)}\right]\breve{\matdel}\right)\\
&= \tr\left(\ident_{\vecn_{2:D}}\breve{\matdel}\right) = \tr\left(\breve{\matdel}\right),
\end{align*}
\begin{align*}
\EV\left[\vecZ_{(1)}\breve{\matdel}\vecZ_{(1)}\T\right] &= \ident_{n_1} \times \tr\left(\breve{\matdel}\right) \\
&= \ident_{n_1} \times \tr\left(\bigotimes_{d=D}^2\matdel_d \right) \\
&= \ident_{n_1} \times \prod_{d=2}^D\tr\left(\matdel_d\right).
\end{align*}
The desired expectation is then,
\begin{align*}
\EV\left[\vecX_{(1)}\vecX_{(1)}\T\right] &= \vecM_{(1)}\vecM_{(1)}\T + \EV\left[W\right]\vecM_{(1)}\vecA_{(1)}\T + \EV\left[W\right]\vecA_{(1)}\vecM_{(1)}\T + \EV\left[W^2\right]\vecA_{(1)}\vecA_{(1)}\T\\
& \qquad +\EV[W]\matdel_1\times\prod_{d=2}^D\tr\left(\matdel_d\right).
\end{align*}
Finally,
\begin{align*}
\EV\left[\vecX_{(1)}\T\vecX_{(1)}\right] &= \EV\left[\left\{\vecM_{(1)} + W\vecA_{(1)} + \sqrt{W}\matdel_1^\half\vecZ_{(1)}\breve{\matdel}^\half\right\}\T \times \left\{\vecM_{(1)} + W\vecA_{(1)} + \sqrt{W}\matdel_1^\half\vecZ_{(1)}\breve{\matdel}^\half\right\}\right]\\
&= \vecM_{(1)}\T\vecM_{(1)} + \EV\left[W\right]\vecM_{(1)}\T\vecA_{(1)} + \EV\left[W\right]\vecA_{(1)}\T\vecM_{(1)} + \EV\left[W^2\right]\vecA_{(1)}\T\vecA_{(1)}\\
& \qquad +\EV[W]\breve{\matdel}^{{\frac{\top}{2}}} \EV\left[\vecZ_{(1)}\T\vecDelta_1\vecZ_{(1)}\right] \breve{\matdel}^{\half}.
\end{align*}

The $\vecZ_{(1)}\T$ term can be broken into row vectors, $\vecz\T_{(1)} \in \mathbb{R}^{1\times n_1}$ and therefore,
\begin{align*}
\EV\left[\vecz_{(1)}\T\matdel_1\vecz_{(1)}\right] &= \EV\left[\tr\left(\vecz_{(1)}\vecz_{(1)}\T\matdel_1\right)\right] \\
&= \tr\left(\EV\left[\vecz_{(1)}\vecz_{(1)}\T\right]\matdel_1\right)\\
&= \tr\left(\ident_{n_1}\matdel_1\right)= \tr\left(\matdel_1\right)\\
\EV\left[\vecZ_{(1)}\T\vecDelta_1\vecZ_{(1)}\right] &= \ident_{\vecn_{2:D}} \times  \tr\left(\matdel_1\right).
\end{align*}
We then arrive at the following expectation,
\begin{align*}
\EV\left[\vecX_{(1)}\T\vecX_{(1)}\right] &= \vecM_{(1)}\T\vecM_{(1)} + \EV\left[W\right]\vecM_{(1)}\T\vecA_{(1)} + \EV\left[W\right]\vecA_{(1)}\T\vecM_{(1)} + \EV\left[W^2\right]\vecA_{(1)}\T\vecA_{(1)}\\
& \qquad + \EV[W]\breve{\matdel}\times  \tr\left(\matdel_1\right).
\end{align*}

\section{Characteristic Function Derivations}
\subsubsection*{TVST}
If $\fX$ follows a tensor variate skew-$t$ distribution with $\nu$ degrees of freedom, then the characteristic function is
\begin{align*}
\mathcal{C}_{\fX}&=\exp\{i\vect'\vecmu\}\int_0^{\infty}\exp\left\{iwa-\frac{1}{2}wb\right\}h(w)dw\\
&=\exp\{i\vect'\vecmu\}\int_0^{\infty}\exp\left\{iwa-\frac{1}{2}wb\right\}\frac{\frac{\nu}{2}^{\frac{\nu}{2}}}{\Gamma\left(\frac{\nu}{2}\right)}w^{-\frac{\nu}{2}-1}\exp\left\{\frac{-\nu}{2w}\right\}dw\\
&=\exp\{i\vect'\vecmu\}\frac{\frac{\nu}{2}^{\frac{\nu}{2}}}{\Gamma\left(\frac{\nu}{2}\right)}\int_0^{\infty}\exp\left\{iwa\right\}w^{-\frac{\nu}{2}-1}\exp\left\{-\frac{1}{2}wb-\frac{\nu}{2w}\right\}dw\\
&=\exp\{i\vect'\vecmu\}\frac{2\frac{\nu}{2}^{\frac{\nu}{2}}K_{-\frac{\nu}{2}}(\sqrt{\nu b})
}{\Gamma\left(\frac{\nu}{2}\right)\left(\frac{b}{\nu}\right)^{-\frac{\nu}{4}}}\mathcal{C}_{\text{GIG}}\left(a~|~b,\nu,-\frac{\nu}{2}\right).
\end{align*}
\subsubsection*{TVGH}
If $\fX$ follows a tensor variate variance gamma distribution with concentration parameter $\omega$ and index parameter $\lambda$ then the characteristic function is
\begin{align*}
\mathcal{C}_{\fX}&=\exp\{i\vect'\vecmu\}\int_0^{\infty}\exp\left\{iwa-\frac{1}{2}wb\right\}\frac{w^{\lambda-1}}{2K_{\lambda}(\omega)}\exp\left\{-\frac{\omega w}{2}-\frac{\omega}{2w}\right\}dw\\
&=\exp\{i\vect'\vecmu\}\frac{1}{2K_{\lambda}(\omega)}\int_0^{\infty}\exp\left\{iwa\right\}w^{\lambda-1}\exp\left\{-\frac{(\omega+b) w}{2}-\frac{\omega}{2w}\right\}dw\\
&=\exp\{i\vect'\vecmu\}\frac{K_{\lambda}(\sqrt{(\omega+b)\omega})}{K_{\lambda}(\omega)\left(\frac{\omega+b}{\omega}\right)^{\frac{\lambda}{2}}}\mathcal{C}_{\text{GIG}}(a~|~\omega+b,\omega,\lambda)\\
\end{align*}

\subsubsection*{TVVG}
If $\fX$ follows a tensor variate variance gamma distribution with concentration parameter $\gamma$ then the characteristic function is
\begin{align*}
\mathcal{C}_{\fX}&=\exp\{i\vect'\vecmu\}\int_0^{\infty}\exp\left\{iwa-\frac{1}{2}wb\right\}\frac{\gamma^{\gamma}}{\Gamma\left(\gamma\right)}w^{\gamma-1}\exp\left\{ -\gamma w \right\}dw\\
&=\exp\{i\vect'\vecmu\}\frac{\gamma^{\gamma}}{\Gamma\left(\gamma\right)}\int_0^{\infty}\exp\left\{iwa\right\}w^{\gamma-1}\exp\left\{ -\left(\gamma+\frac{1}{2}b\right) w \right\}dw\\
&=\exp\{i\vect'\vecmu\}\frac{\gamma^{\gamma}}{\left(\gamma+\frac{1}{2}b\right)^{\gamma}}\mathcal{C}_{\text{Gamma}}\left(a~\middle|~\gamma,\gamma+\frac{1}{2}b\right),\\
\end{align*}
where 
$$
\mathcal{C}_{\text{Gamma}}\left(a~\middle|~\gamma,\gamma+\frac{1}{2}b\right)=\left(1-\frac{2ia}{2\gamma+b}\right)^{-\gamma}
$$
is the characteristic function of a gamma distribution with parameters $\gamma$ and $\gamma+\frac{1}{2}b$ evaluated at $a$. 

\subsubsection*{TVNIG}
If $\fX$ follows a tensor variate variance gamma distribution with concentration parameter $\kappa$ then the characteristic function is
\begin{align*}
\mathcal{C}_{\fX}&=\exp\{i\vect'\vecmu\}\int_0^{\infty}\exp\left\{iwa-\frac{1}{2}wb\right\}\frac{\exp\{\kappa\}w^{-\frac{2}{3}}}{\sqrt{2\pi}}\exp\left\{-\frac{1}{2}\left(\kappa^2w+\frac{1}{w}\right)\right\}dw\\
&=\exp\{i\vect'\vecmu\}\frac{\exp\{\kappa\}}{\exp\{\sqrt{\kappa^2+b}\}}\int_0^{\infty}\exp\left\{iwa\right\}\frac{\exp\{\sqrt{\kappa^2+b}\}w^{-\frac{2}{3}}}{\sqrt{2\pi}}\exp\left\{-\frac{1}{2}\left((\kappa^2+b)w+\frac{1}{w}\right)\right\}dw\\
&=\exp\{i\vect'\vecmu\}\frac{\exp\{\kappa\}}{\exp\{\sqrt{\kappa^2+b}\}}\mathcal{C}_{\text{IG}}(a~|~1,\sqrt{\kappa^2+b}),
\end{align*}
where
$$
\mathcal{C}_{\text{IG}}(t~|~1,\gamma)=\exp\left\{\gamma\left(1-\sqrt{1-\frac{2it}{\gamma^2}}\right)\right\}
$$
is the characteristic function of the $\text{IG}(1,\gamma)$ distribution.

\section{E-Step Distributions}
Recall, that we need to find conditional expectations in the E-step of the ECM algorithm. The conditional distributions each follow a GIG distribution with specific parameters shown below. 
\begin{align*}
W_{i}^{\text{ST}}~|~\vecX_i&\sim \text{GIG}\left(\rho\left(\tenA,\bigotimes_{d=1}^D\matdel_{d}^{-1}\right),\delta\left(\tenX;\tenM,\bigotimes_{d=1}^D\matdel_{d}^{-1}\right)+\nu,-(\nu+n^*)/2\right),\\
W_{i}^{\text{GH}}~|~\vecX_i&\sim \text{GIG}\left(\rho\left(\tenA,\bigotimes_{d=1}^D\matdel_{d}^{-1}\right)+\omega,\delta\left(\tenX;\tenM,\bigotimes_{d=1}^D\matdel_{d}^{-1}\right)+\omega,\lambda-{n^*}/{2}\right),\\
W_{i}^{\text{VG}}~|~\vecX_i&\sim \text{GIG}\left(\rho\left(\tenA,\bigotimes_{d=1}^D\matdel_{d}^{-1}\right)+2\gamma,\delta\left(\tenX;\tenM,\bigotimes_{d=1}^D\matdel_{d}^{-1}\right),\gamma-{n^*}/{2}\right),\\
W_{i}^{\text{NIG}}~|~\vecX_i,&\sim \text{GIG}\left(\rho\left(\tenA,\bigotimes_{d=1}^D\matdel_{d}^{-1}\right)+\kappa^2,\delta\left(\tenX;\tenM,\bigotimes_{d=1}^D\matdel_{d}^{-1}\right)+1,-{(1+n^*)}/{2}\right).
\end{align*}
Fortunately, the expectations of functions of a GIG random variable we need for the E-step can be found in a mathematically tractable form. Specifically if $Y\sim \text{GIG}(a,b,\lambda)$ then we have from \cite{mcnicholas16a} that

\begin{equation}
\mathbb{E}(Y)=\sqrt{\frac{b}{a}}\frac{K_{\lambda+1}(\sqrt{ab})}{K_{\lambda}(\sqrt{ab})},
\label{eq:EY}
\end{equation}

\begin{equation}
\mathbb{E}\left(\frac{1}{Y}\right)=\sqrt{\frac{a}{b}}\frac{K_{\lambda+1}(\sqrt{ab})}{K_{\lambda}(\sqrt{ab})}-\frac{2\lambda}{b},
\end{equation}

\begin{equation}\label{eq:elogy}
\mathbb{E}(\log(Y))=\log\left(\sqrt{\frac{b}{a}}\right)+\frac{1}{K_{\lambda}(\sqrt{ab})}\frac{\partial}{\partial \lambda}K_{\lambda}(\sqrt{ab}).
\end{equation}

\section{Updates for the Additional Parameters}
\subsection*{TVST}
In the case of the TVST distribution, the degrees of freedom $\nu$ needs to be updated. The update for the degrees of freedom cannot be obtained in closed form. Instead we solve Equation \eqref{eq:nuup} for $\nu$ to obtain $\hat{\nu}^{(t+1)}$. 

\begin{equation}
\log\left(\frac{\nu}{2}\right)+1-\varphi\left(\frac{\nu}{2}\right)-\frac{1}{N}\sum_{i=1}^N(b_i+c_i)=0,
\label{eq:nuup}
\end{equation}
where $\varphi(\cdot)$ is the digamma function.
\subsection*{TVGH}
In the case of the TVGH distribution, we would update $\lambda$ and $\omega$. In this case,
\begin{equation}
\mathcal{L}_1=N\log(K_{\lambda}(\omega))-\lambda\sum_{i=1}^Nc_i-\frac{1}{2}\omega\sum_{i=1}^N\left(a_i+b_i\right)
\label{eq:L1GH}
\end{equation}

The updates for $\lambda$ and $\omega$ cannot be obtained in closed form. However, \cite{browne15} discuss numerical methods for these updates, and therefore because the portion of the likelihood function that include these parameters is the same as in the multivariate case, the updates described in \cite{browne15} can be used directly here.

The updates for $\lambda$ and $\omega$ rely on the log convexity of $K_{\lambda}(\omega)$, \cite{baricz10}, in both $\lambda$ and $\omega$ and maximizing \eqref{eq:L1GH} via conditional maximization.
The resulting updates are
\begin{align}
\hat{\lambda}^{(t+1)}&=\overline{c}\hat{\lambda}^{(t)}\left[\left.\frac{\partial}{\partial s}\log(K_{s}(\hat{\omega}^{(t)}))\right|_{s=\hat{\lambda}^{(t)}}\right]^{-1} \label{eq:lamup}\\
\hat{\omega}^{(t+1)}&=\hat{\omega}^{(t)}-\left[\left.\frac{\partial}{\partial s}q(\hat{\lambda}^{(t+1)},s)\right|_{s=\hat{\omega}^{(t)}}\right]\left[\left.\frac{\partial^2}{\partial s^2}q(\hat{\lambda}^{(t+1)},s)\right|_{s=\hat{\omega}^{(t)}}\right]^{-1} \label{eq:omup}
\end{align}
where the derivative in \eqref{eq:lamup} is calculated numerically and $\overline{c}=\sum_{i=1}^Nc_i/N$. The partials in \eqref{eq:omup} are described in \cite{browne15}, and can be written as
$$
\frac{\partial}{\partial \omega}q(\lambda,\omega)=\frac{1}{2}[R_{\lambda}(\omega)+R_{-\lambda}(\omega)-(\overline{a}+\overline{b})],
$$
and 
$$
\frac{\partial^2}{\partial \omega^2}q(\lambda,\omega)=\frac{1}{2}\left[R_{\lambda}(\omega)^2-\frac{1+2\lambda}{\omega}R_{\lambda}(\omega)-1+R_{-\lambda}(\omega)^2-\frac{1-2\lambda}{\omega}R_{-\lambda}(\omega)-1\right],
$$
where $R_{\lambda}(\omega)=K_{\lambda+1}(\omega)/K_{\lambda}(\omega)$.

\subsubsection*{TVVG}
In the case of the TVVG, the update for $\gamma$ is needed. This update, like the TVST and TVGH, cannot be obtained in closed form. Instead, the update, $\gamma^{(t+1)}$, is obtained by solving \eqref{eq:gammup} for $\gamma$.
\begin{equation}
\log(\gamma)+1-\varphi(\gamma)+\overline{c}-\overline{a}=0.
\label{eq:gammup}
\end{equation}

\subsubsection*{TVNIG}
Finally, in the TVNIG case, the update for $\kappa$ can be written in closed form as 
$$
\hat{\kappa}=\frac{N}{\sum_{i=1}^Na_{i}}.
$$

\section{Computational Considerations}\label{sec:app_sw}

Singular $\vecDelta_{g}$ values were numerically regularized by adding a small positive quantity to the diagonal elements of the matrices \citep{williams2006}. The regularization is summarized in the following equation:
\begin{equation}\label{eq:covreg}
\tilde{\vecDelta} = \hat{\vecDelta} + \epsilon\ident,
\end{equation}
where $\epsilon \in (0,0.1]$, $\hat{\vecDelta}$ is the estimated singular scale matrix, and $\tilde{\vecDelta}$ is the regularized estimate of $\vecDelta$. We used $\epsilon = 0.001$ in our implementation. The singularity of $\hat{\vecDelta}$ was assessed by checking if its inverse condition number is less than machine epsilon. This regularization is often done implicitly in software implementations such as {\tt scikit-learn}'s GaussianMixture function, written in Python. The value of the regularization parameter $\epsilon$ could be tuned. The larger it is, the further the model results are from the true solution. The positive definiteness of the $\vecDelta_d$ matrices was checked using the Cholesky decomposition. 

To stop our ECM algorithms, we use a criterion based on the Aitken acceleration \citep{aitken26}. At iteration $t$ of the ECM algorithm, the Aitken acceleration is 
\begin{equation}
a^{(t)}=\frac{l^{(t+1)}-l^{(t)}}{l^{(t)}-l^{(t-1)}},
\end{equation}
where $l^{(t)}$ is the (observed) log-likelihood at iteration $t$. \cite{bohning94} use $a^{(t)}$ to calculate an asymptotic estimate of the log-likelihood at iteration $t+1$:
\begin{equation}
l^{(t+1)}_\infty=l^{(t)}+\frac{1}{1-a^{(t)}}(l^{(t+1)}-l^{(t)}).
\end{equation}
We stop the EM algorithm when $l^{(t+1)}_\infty-l^{(t)}<\epsilon$ \citep{mcnicholas10a}.

All the figures are made using Julia, using version 1.3.1 of the {\tt Gadfly} visualization package \citep{gadfly}. 
Bessel function values are calculated using 100 digit numbers, made possible by version 1.2.4 of the {\tt ArbNumerics.jl} library. We use numerical differentiation to find $\frac{\partial}{\partial \lambda}K_{\lambda}(\sqrt{ab})$ in Equation \ref{eq:elogy}.

\section{Simulation}\label{sec:app_sim}
\subsection{Normal Data}

Following Definition 2.2 in \cite{ohlson2013}, the normal data is generated using the following equation:  
\begin{equation}\label{eq:sim_mln}
\vecc(\tenX) = \vecc(\tenM) + \bigotimes_{d=1}^D\vecDelta_d^{\frac{1}{2}}\vecu \text{,}
\end{equation}
where $\vecu$ is a vector of iid $\mathcal{N}(0,1)$ random numbers. This is equivalent to the multivariate normal model for the vectorized version of the tensor data.  We can generate $\vecu$ as a mode $d$ tensor and use tensor $d$-mode products to implement the final term on the right hand side of \eqref{eq:sim_mln} \citep{kolda2009}.  It has the advantage of retaining the tensor structure of the data and not creating one large matrix from the Kronecker product, $\bigotimes_{d=1}^D\vecDelta_d^{\frac{1}{2}}$ and then having to permute the data back into a tensor format.


A signal-to-noise ratio of one half was applied to the simulated data prior to analysis. The $\vecDelta_d$ parameters were generated by specifying a diagonal matrix of eigenvalues and a random orthogonal matrix and combining them as you would in an eigen-decomposition of the scale matrix. An $n_d \times n_d$ orthogonal matrix was created by generating $n_d^2$ iid $N(0,1)$ random values, placing them in a matrix and orthogonalizing it with the QR decomposition. We restrict the condition number of these $\vecDelta_d$ matrices to be at most 10.


\begin{figure}[!htb]
	\begin{center}
		\includegraphics[height=0.5\textwidth]{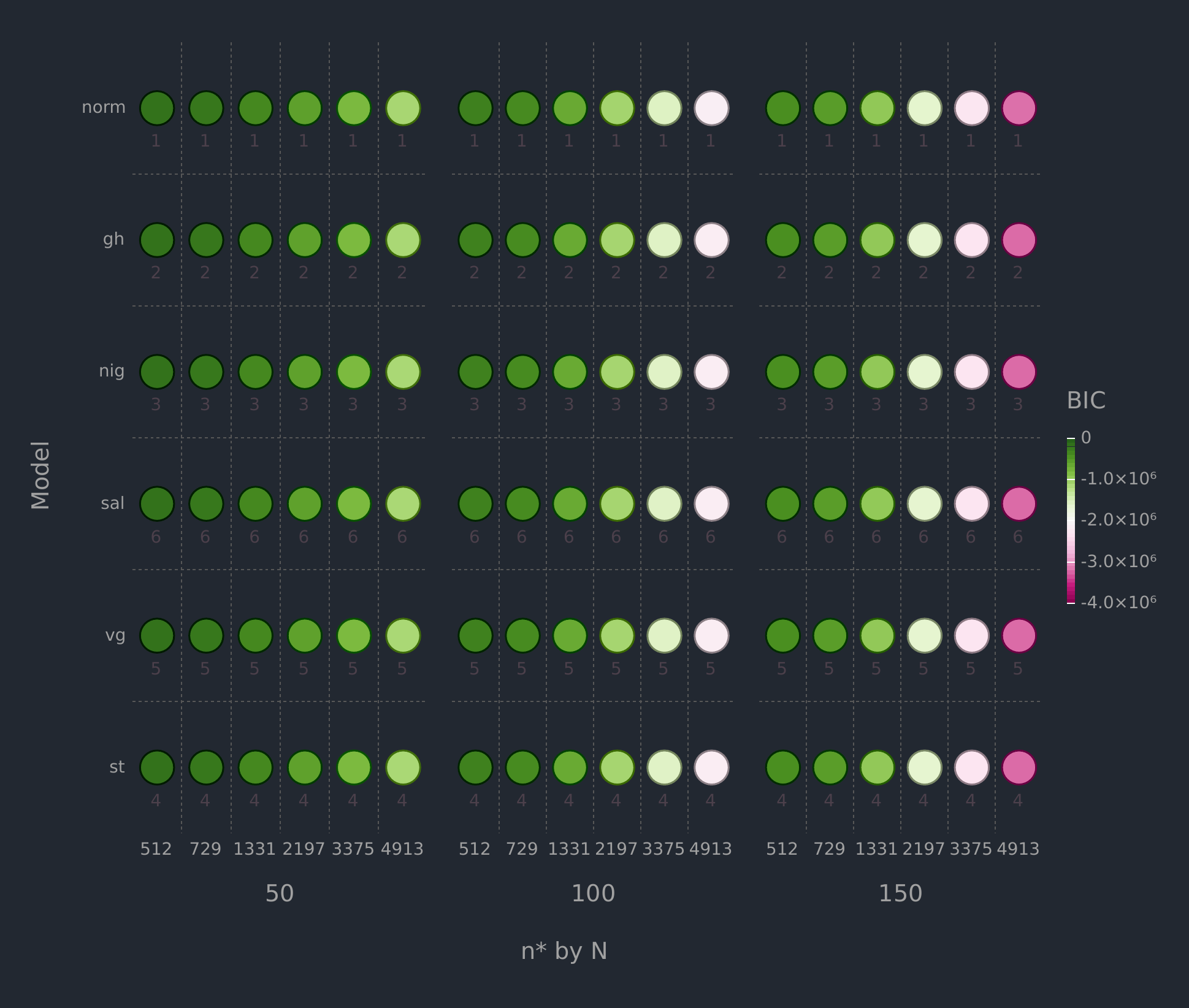}
		\caption{Average BIC and rank of the models for each combination of $n^*$ and $N$ for the normal data simulation.}
		\label{fig:sim_norm_bic}
	\end{center}
\end{figure}  

Figure \ref{fig:sim_norm_bic} summarizes the average BIC values and the rank of the models for the 100 simulations, across the values of $N$ and $n^*$ for each of the models. Based on their BIC values, the tensor normal is consistently the top performer, despite doing a poor job estimating the scale matrices. Of the skewed models, the TVGH and TVNIG models ranked highest. In Figures \ref{fig:sim_norm_ecdf_m}, empirical distribution plots of the relative error of the mode-1 matricization of $\mathbb{E}[\fX]$. We see that the normal, TVGH and TVNIG have very short tails, whereas the tails are longer for the other distributions, indicating higher variation in the error. The results for the TVST, TVVG and TVSAL are all influenced by $n^*$, were larger tensors result in elevated relative errors. Figure  \ref{fig:sim_norm_ecdf_kp} shows a similar plot for the error of the Kronecker product of the scale matrices. In this case, the empirical distributions for all of the skewed distributions have very short tails, whereas the tails are longer for the normal distribution. Tensor size has a small effect on each models relative error. 

\begin{figure}[!htb]
	\begin{center}
		\includegraphics[height=0.5\textwidth]{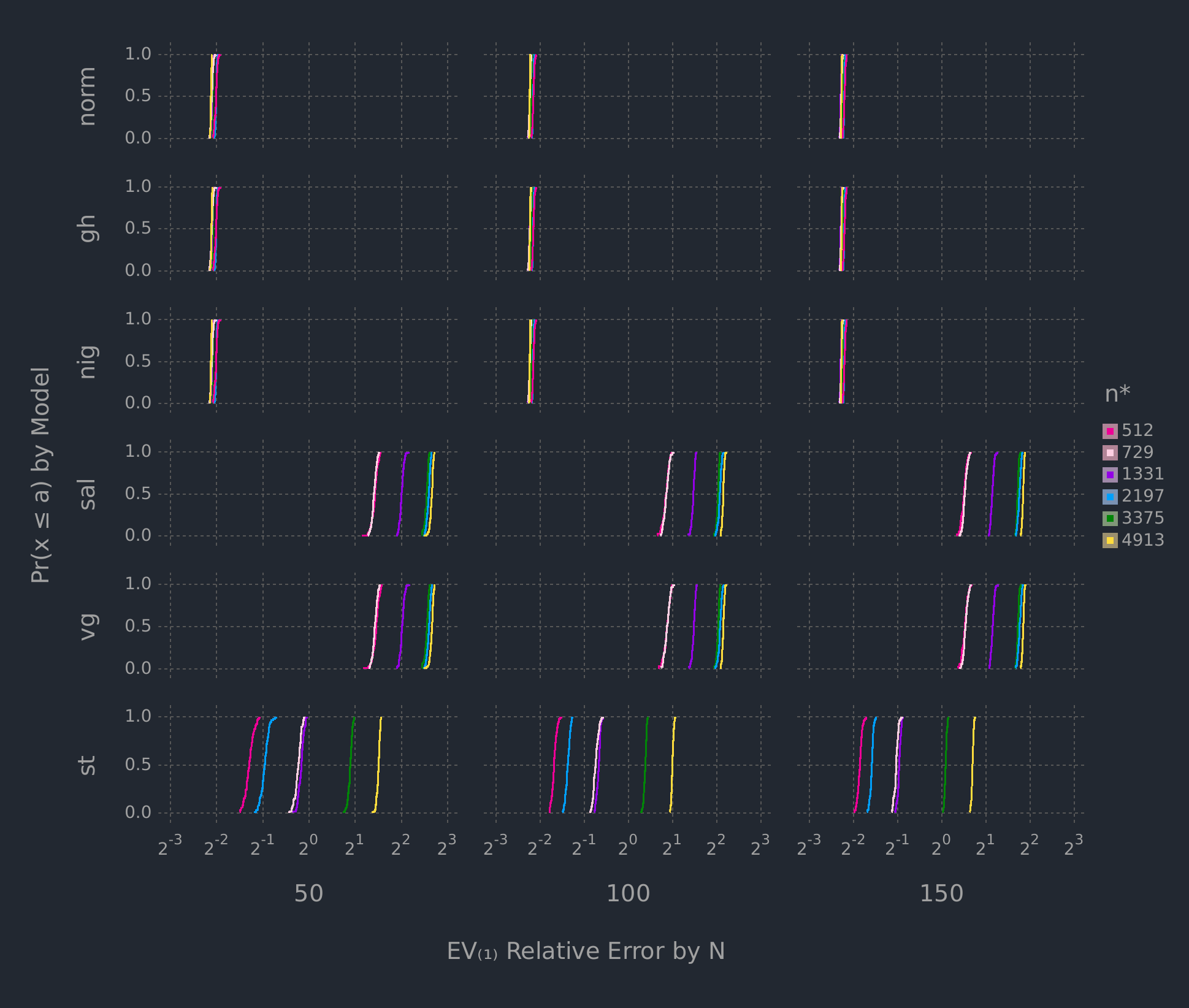}
		\caption{Empirical distribution plots of the relative error in the mode-1 matricization of $\mathbb{E}[\fX]$ for the normal data simulation.}
		\label{fig:sim_norm_ecdf_m}
	\end{center}
\end{figure}  

\begin{figure}[!htb]
	\begin{center}
		\includegraphics[height=0.5\textwidth]{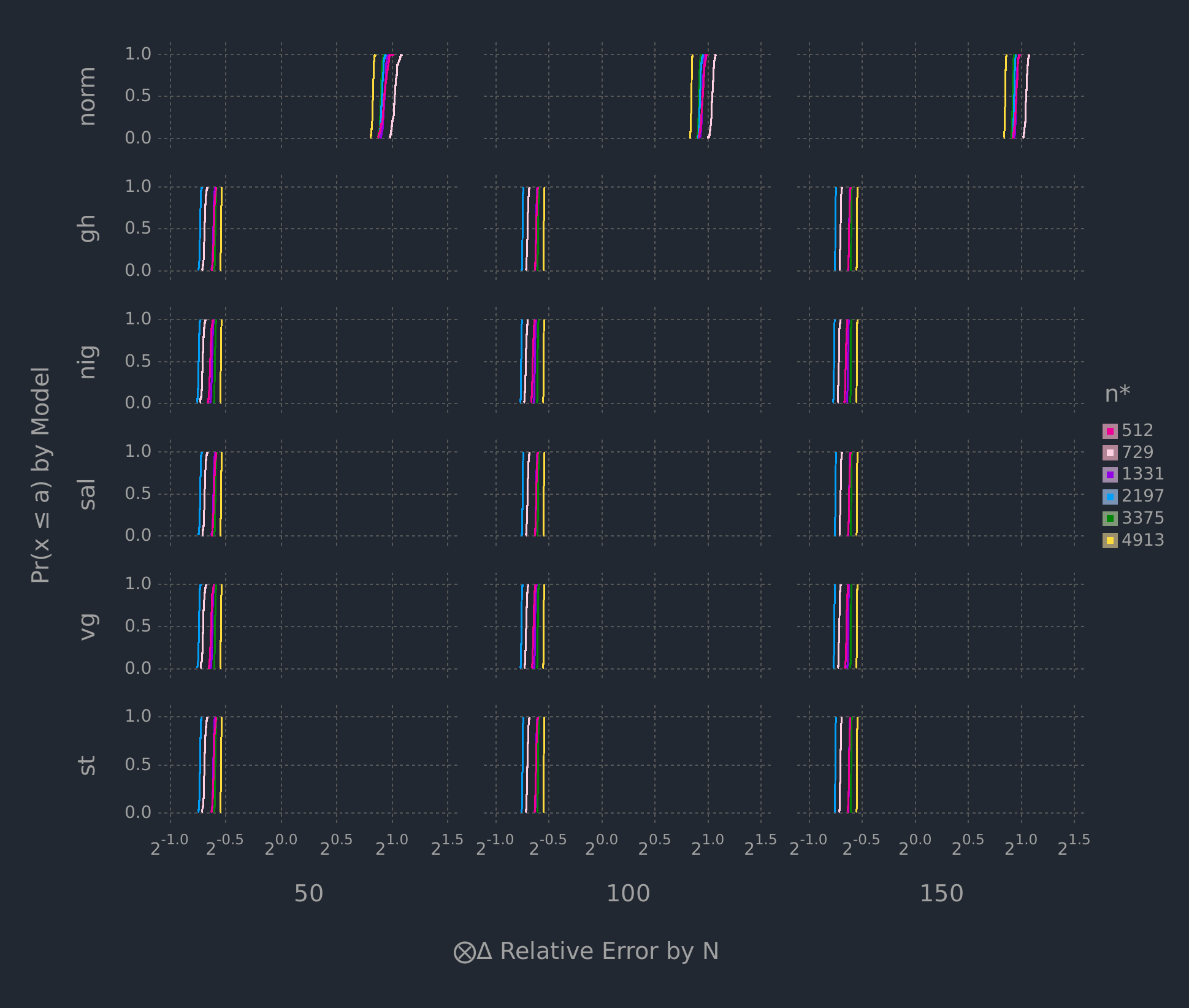}
		\caption{Empirical distribution plots of the relative error in $\bigotimes_{d=1}^D\vecDelta_d$ for the normal data simulation.}
		\label{fig:sim_norm_ecdf_kp}
	\end{center}
\end{figure}  

\subsection{Skewed Data}

We now consider the simulation study for which the data was simulated form the TVST distribution. The scale matrices were generated the same way as the normal data and they were combined as part of $\tenV$ using tensor $d$-mode products. We used a signal to noise ratio of one half. 

\begin{figure}[!htb]
	\begin{center}
		\includegraphics[height=0.5\textwidth]{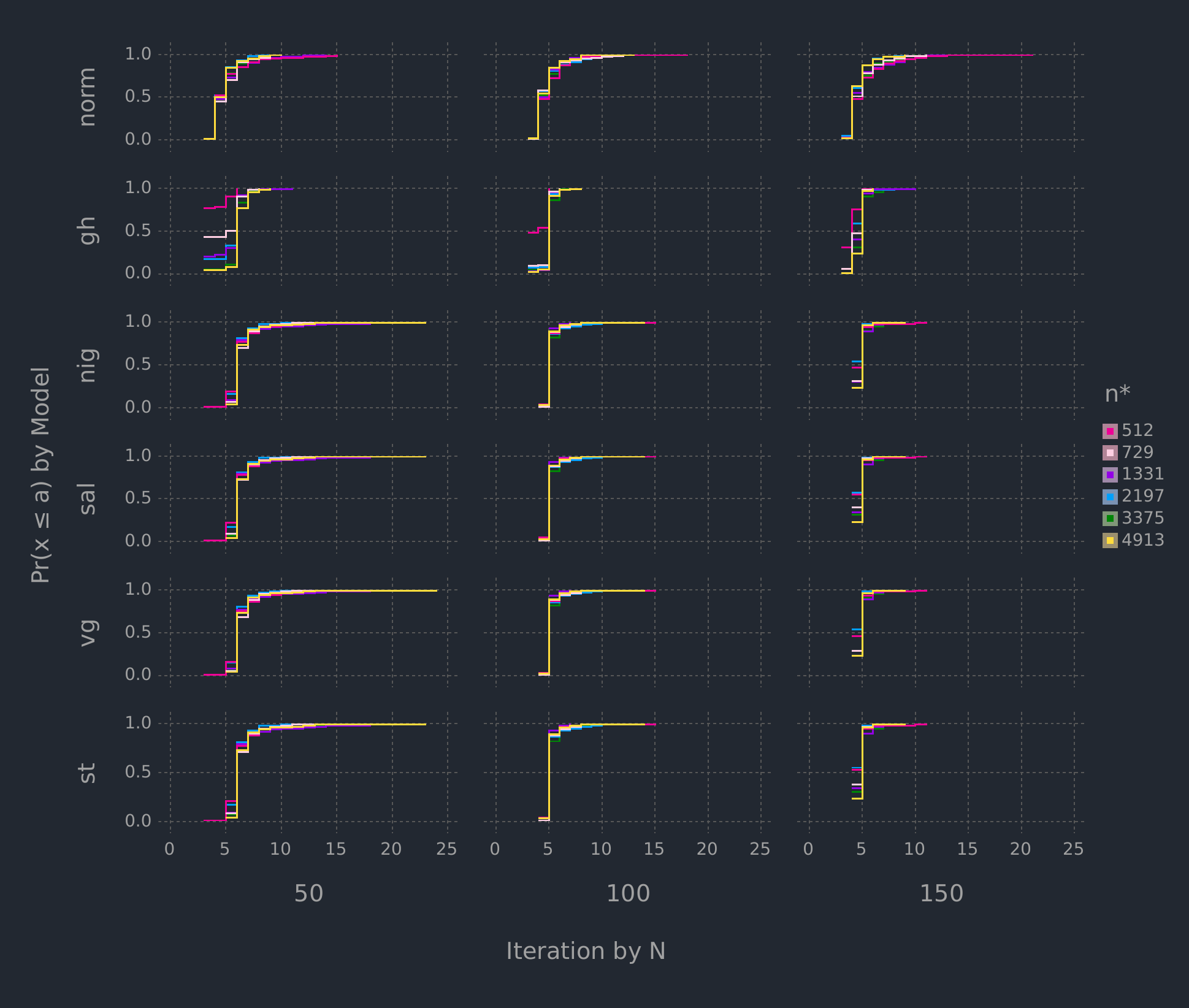}
		\caption{Number of iterations for each model by each combination of $N$ and $n^*$ for the skewed data simulation.}
		\label{fig:sim_skew_it}
	\end{center}
\end{figure}  

Figure \ref{fig:sim_skew_it}, which shows the empirical distribution for the number of iterations, indicates that the normal model had long tails for all values of $N$ and $n^*$, with the longest tails occurring for the smallest tensors ($n^*=512$). The tails for the skewed distributions decrease as the sample size increases. Aside from the TVGH, the values of $n^*$ do not affect the distribution of iterations.

\begin{figure}[!htb]
	\begin{center}
		\includegraphics[height=0.5\textwidth]{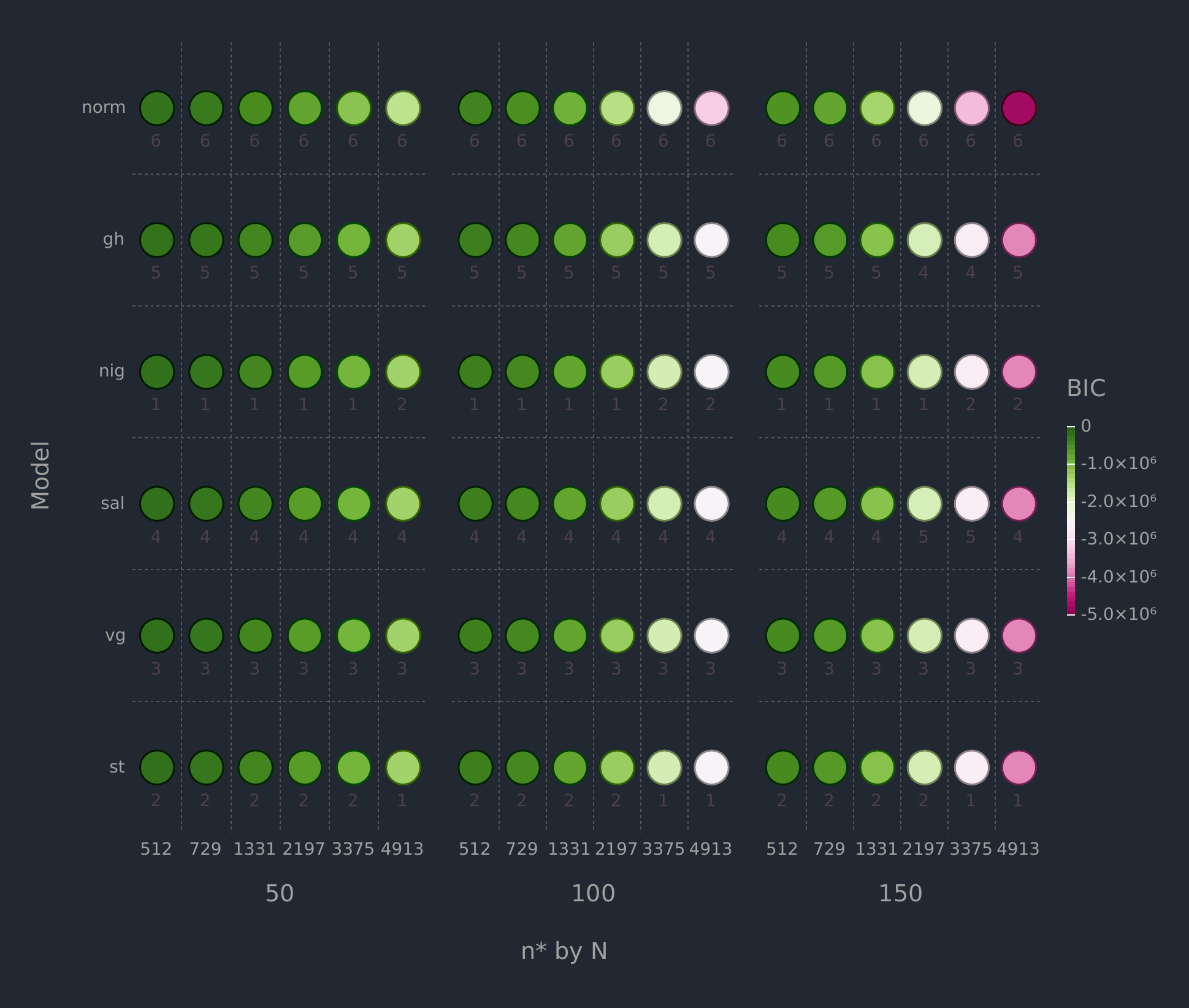}
		\caption{Average BIC and rank of the models for each combination of $N$ and $n^*$.}
		\label{fig:sim_skew_bic}
	\end{center}
\end{figure}  

Figure \ref{fig:sim_skew_bic}, which shows a similar plot to Figure \ref{fig:sim_norm_bic},  indicates the normal distribution is consistently the poorest performer among the models. For small to moderate sized tensors, the TVNIG model consistently ranks the highest. As $n^*$ reaches its maximum size, the TVST model overtakes the TVNIG model in the rankings. 

\begin{figure}[!htb]
	\begin{center}
		\includegraphics[height=0.5\textwidth]{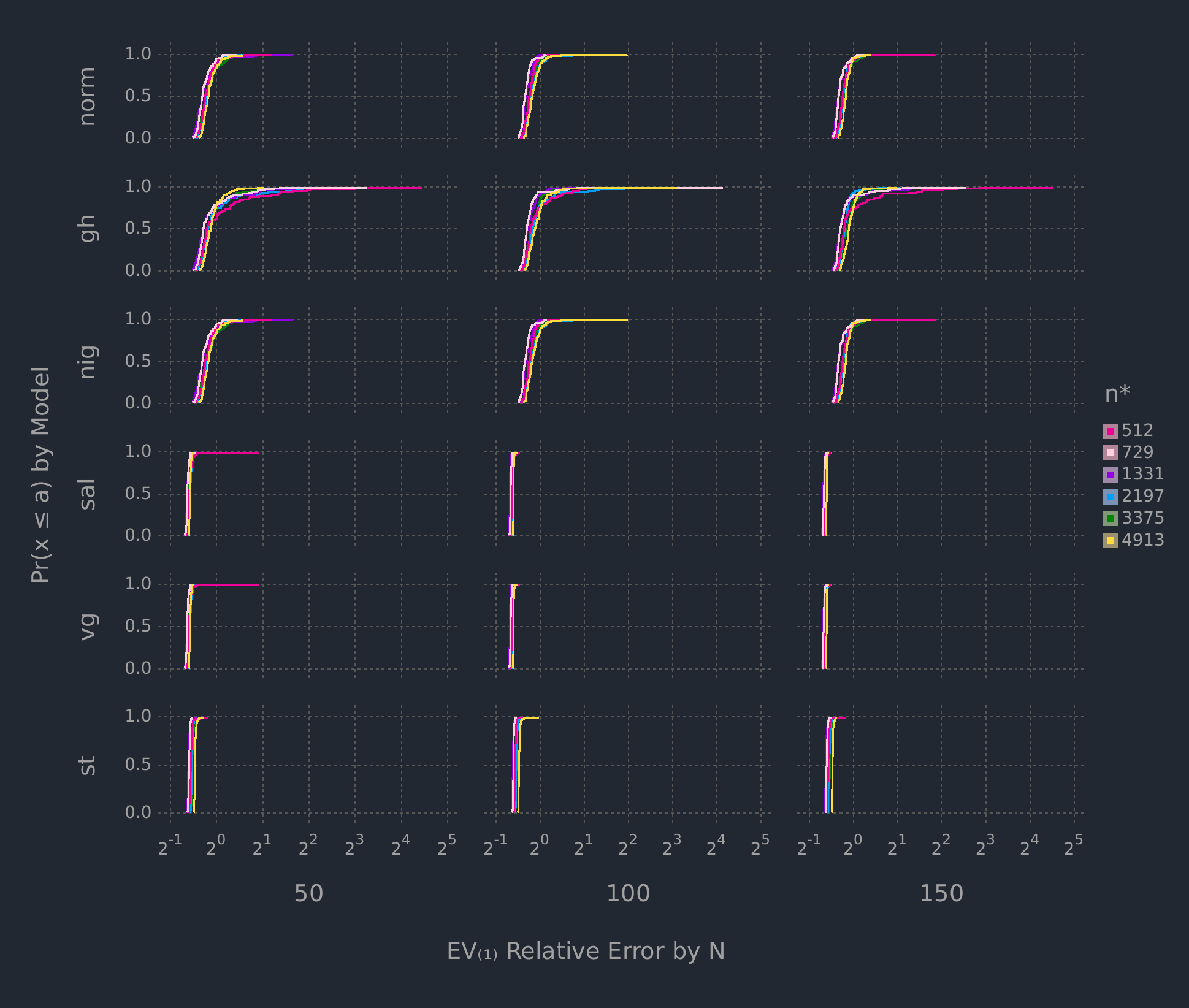}
		\caption{Empirical distribution plots of the relative error in the mode 1 matricization of $\mathbb{E}[\fX]$ for the skewed data simulation.}
		\label{fig:sim_skew_ecdf_m}
	\end{center}
\end{figure}  

The empirical distribution plots of the relative error of the mode-1 matricization of $\mathbb{E}[\fX]$ are plotted in Figure~\ref{fig:sim_skew_ecdf_m}. The Normal, TVGH and TVNIG models all have long right tails, irrespective of the size of $N$ and $n^*$. At small $N$ and $n^*$, the TVVG and TVSAL both have long right tails. Like the TVST, these tails are not present as both $N$ and $n^*$ increase.

\begin{figure}[!htb]
	\begin{center}
		\includegraphics[height=0.5\textwidth]{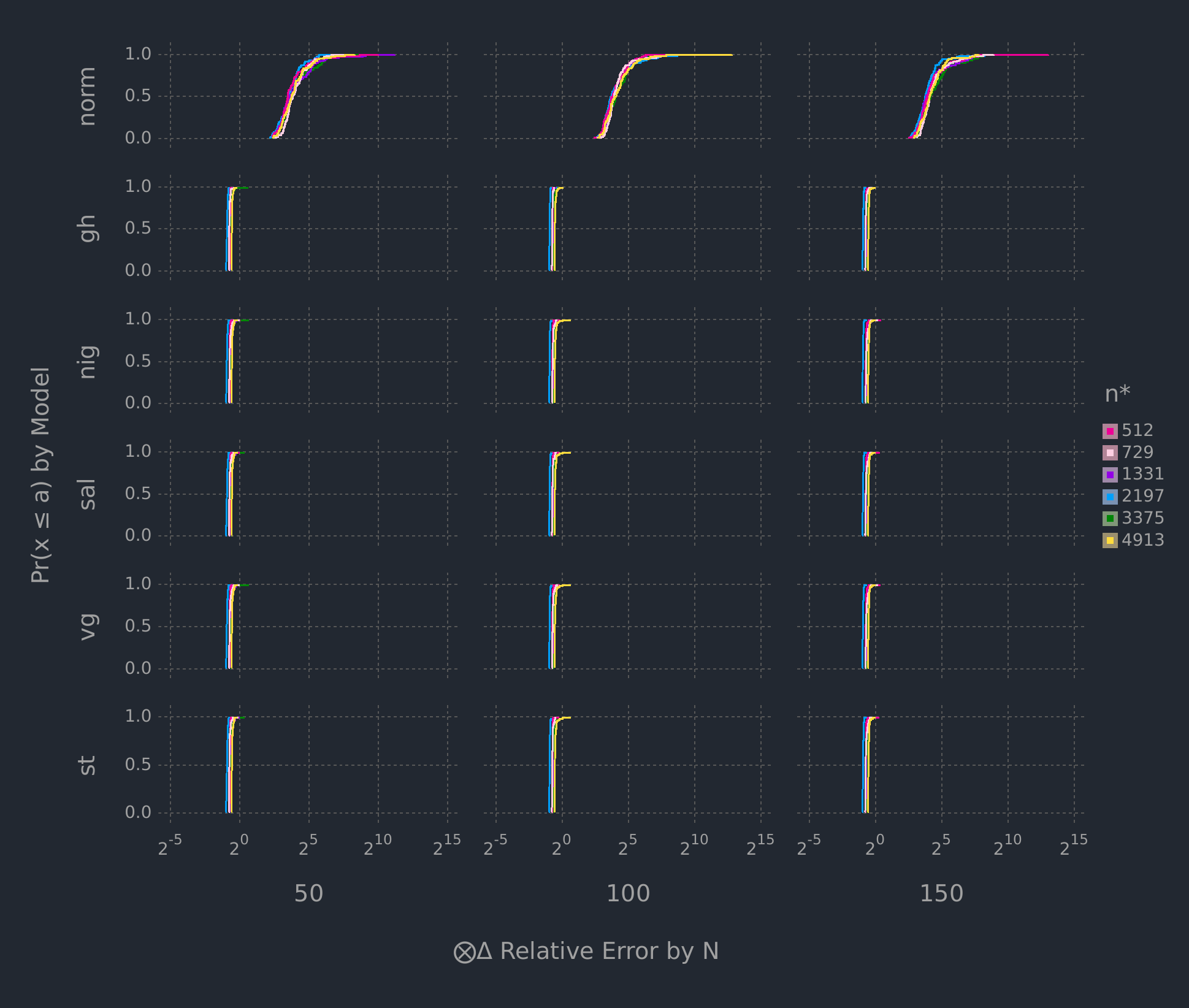}
		\caption{Empirical distribution plots of the relative error in $\bigotimes_{d=1}^D\vecDelta_d$ for the skewed data simulation.}
		\label{fig:sim_skew_ecdf_kp}
	\end{center}
\end{figure}  

Figure~\ref{fig:sim_skew_ecdf_kp} displays the distribution of relative errors for the Kronecker product of the scale matrices. The skewed distributions all have relative errors below 1 and are not influenced by $N$ or $n^*$. The Normal model performs very poorly across the range of $N$ and $n^*$ values, exhibiting very long right tails and median relative error values of $\sim 10$.

%
%

\begin{figure}[!htb]
	\begin{center}
		\includegraphics[height=0.5\textwidth]{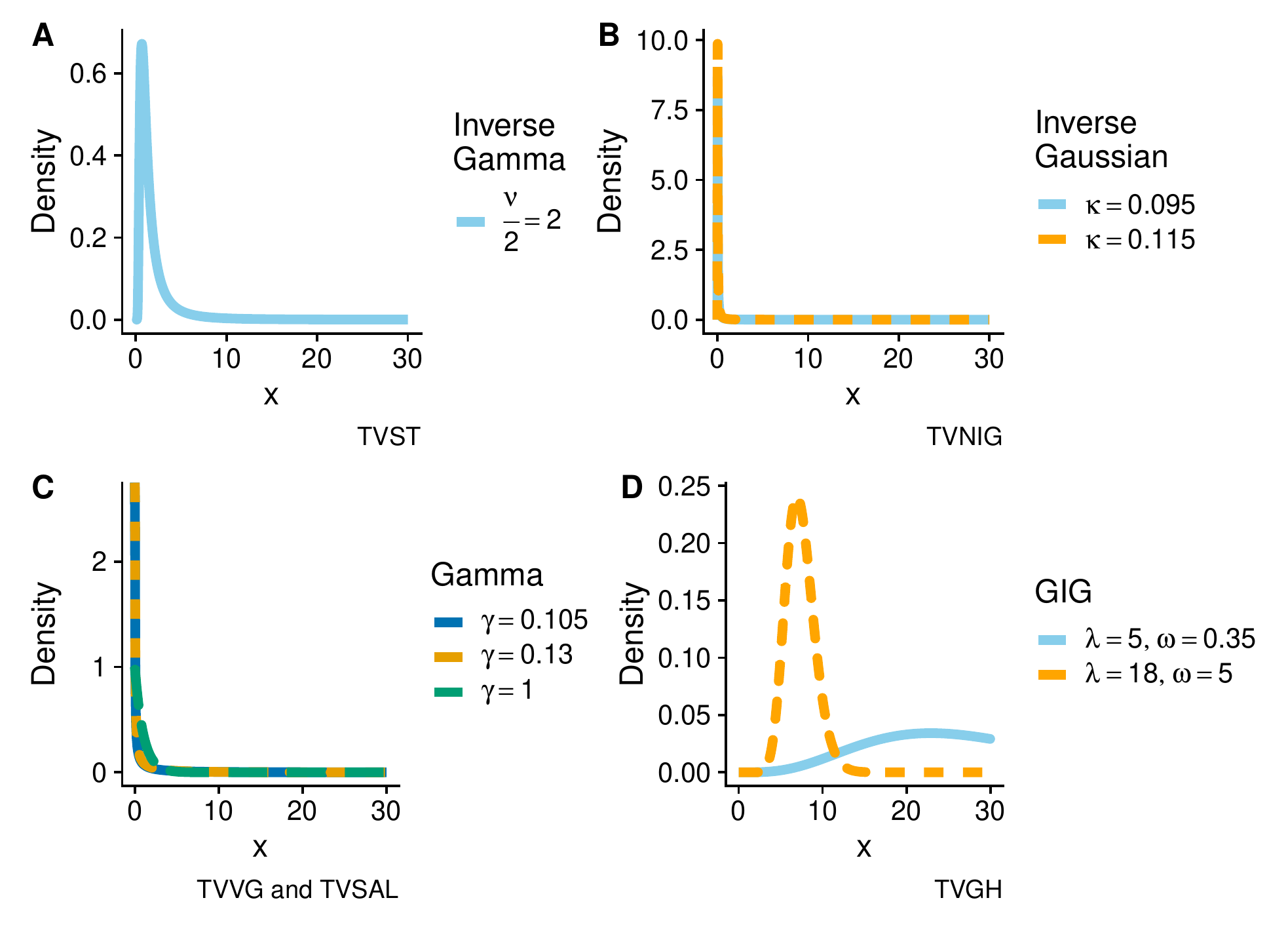}
		\caption{The underlying distributions of $W_{i,g}$ for each of the five tensor variate distributions.}
		\label{fig:sim_w_dist}
	\end{center}
\end{figure} 

The variability in the TVGH $\mathbb{E}[\fX]$ results, visualized in Figures \ref{fig:sim_skew_m} and \ref{fig:sim_skew_ecdf_m}, can be explained by the array of GIG parameter values learned from the data. Each of the underlying distributions of $W_{i,g}$ are visualized in figure \ref{fig:sim_w_dist}. Subplot A represents the distribution that was used to generate the simulated data.  Subplots B to D represent the smallest and largest value(s) of the $W_{i,g}$ distribution parameters seen in the simulations. The TVNIG, TVVG and TVSAL models are learning parameterizations that create densities resembling the inverse gamma density in subplot A. The shapes of the GIG distributions in Subplot D vary considerably, often looking nothing like the distribution in Subplot A. 

\begin{figure}[!htb]
	\begin{center}
		\includegraphics[height=0.5\textwidth]{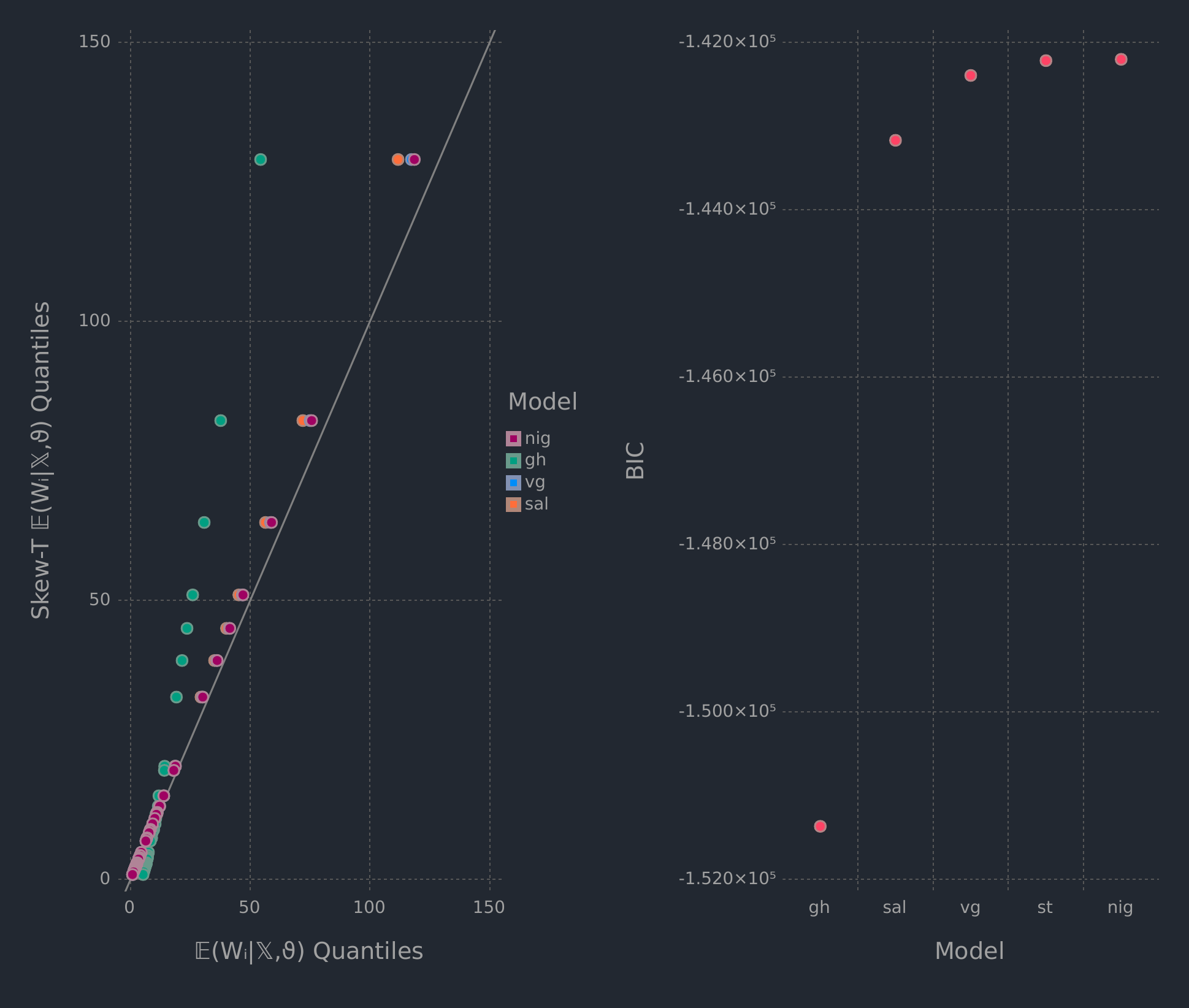}
		\caption{Left Panel:Quantile-quantile plot of $\mathbb{E}(W_{i}~|~\tenX_i,\hat{\bvtheta})$ from the $W_{i,g}$ distributions. The GIG for the TVGH distribution has $\omega = 0.34$ and $\lambda = 4.99$. These $\mathbb{E}(W_{i}~|~\tenX_i,\hat{\bvtheta})$ values directly impact the model parameter values (e.g.$\mathbb{E}[\fX]$). Right Panel: BIC values for each of the skewed models.}
		\label{fig:sim_ey1}
	\end{center}
\end{figure}  
\begin{figure}[!htb]
	\begin{center}
		\includegraphics[height=0.5\textwidth]{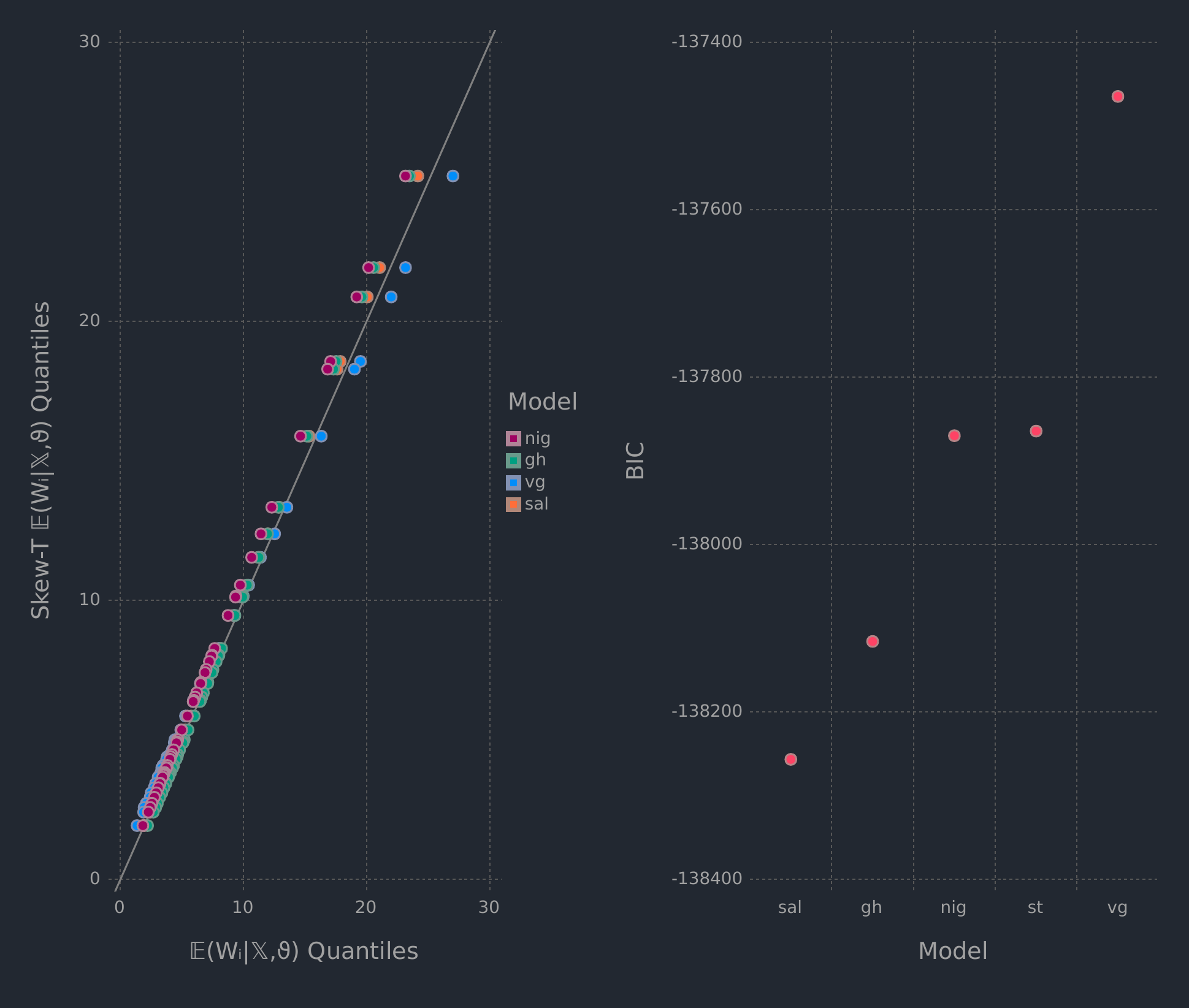}
		\caption{Left Panel: Quantile-quantile plot of $\mathbb{E}(W_{i}~|~\tenX_i,\hat{\bvtheta})$ from the $W_{i,g}$ distributions. The GIG for the TVGH distribution has $\omega = 4.92$ and $\lambda = 17.95$. These $\mathbb{E}(W_{i}~|~\tenX_i,\hat{\bvtheta})$ values directly impact the model parameter values (e.g.$\mathbb{E}[\fX]$). Right Panel: BIC values for each of the skewed models.}
		\label{fig:sim_ey2}
	\end{center}
\end{figure}  

This dissimilarity of the $W_{i,g}$ distributions between the models is reflected in the $\mathbb{E}(W_{i}~|~\tenX_i,\hat{\bvtheta})$ values learned from the data. Starting with two data sets from the skewed simulated data, where $n^*=512$ and $N=50$, we visualize the distribution of the $\mathbb{E}(W_{i}~|~\tenX_i,\hat{\bvtheta})$ values and the resulting model performance in figures \ref{fig:sim_ey1} and \ref{fig:sim_ey2}. The left hand panel of figure \ref{fig:sim_ey1} uses a qq-plot to visualize the distribution of the TVST $\mathbb{E}(W_{i}~|~\tenX_i,\hat{\bvtheta})$ values verses the distribution of the $\mathbb{E}(W_{i}~|~\tenX_i,\hat{\bvtheta})$ values from the other four tensor variate distributions. Recall that the data was generated from a TVST distribution with $\nu = 4$. The right hand panel includes the model BIC values. The distribution of the $W_{i,g}$'s from the TVGH model resembles the blue curve in figure \ref{fig:sim_w_dist} subplot D. This results in values of $\mathbb{E}(W_{i}~|~\tenX_i,\hat{\bvtheta})$ that are divergent from the TVST values and ultimately, in very poor relative model performance, as measured by BIC. The   poor relative performance is due to the effect the $\mathbb{E}(W_{i}~|~\tenX_i,\hat{\bvtheta})$ values have on the model parameter values (e.g.$\mathbb{E}[\fX]$). Contrast this with the results in figure \ref{fig:sim_ey2} where the distribution of $W_{i,g}$'s from the TVGH model is more akin to the distribution used to generate the simulated data. In this instance, the qq-plot indicates the distribution of the $\mathbb{E}(W_{i}~|~\tenX_i,\hat{\bvtheta})$ values between the TVST model and the other models is similar and the model performance between the 5 models is very comparable.

\section{Image Analysis}\label{sec:app_img}
\begin{figure}[!htb]
	\begin{center}
		\includegraphics[height=0.45\textwidth]{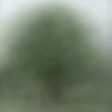}
		\caption{An image of the location tensor, $\tenM$, from the NIG model for the image data.}
		\label{fig:img_loc}
	\end{center}
\end{figure}  
\begin{figure}[!htb]
	\begin{center}
		\includegraphics[height=0.45\textwidth]{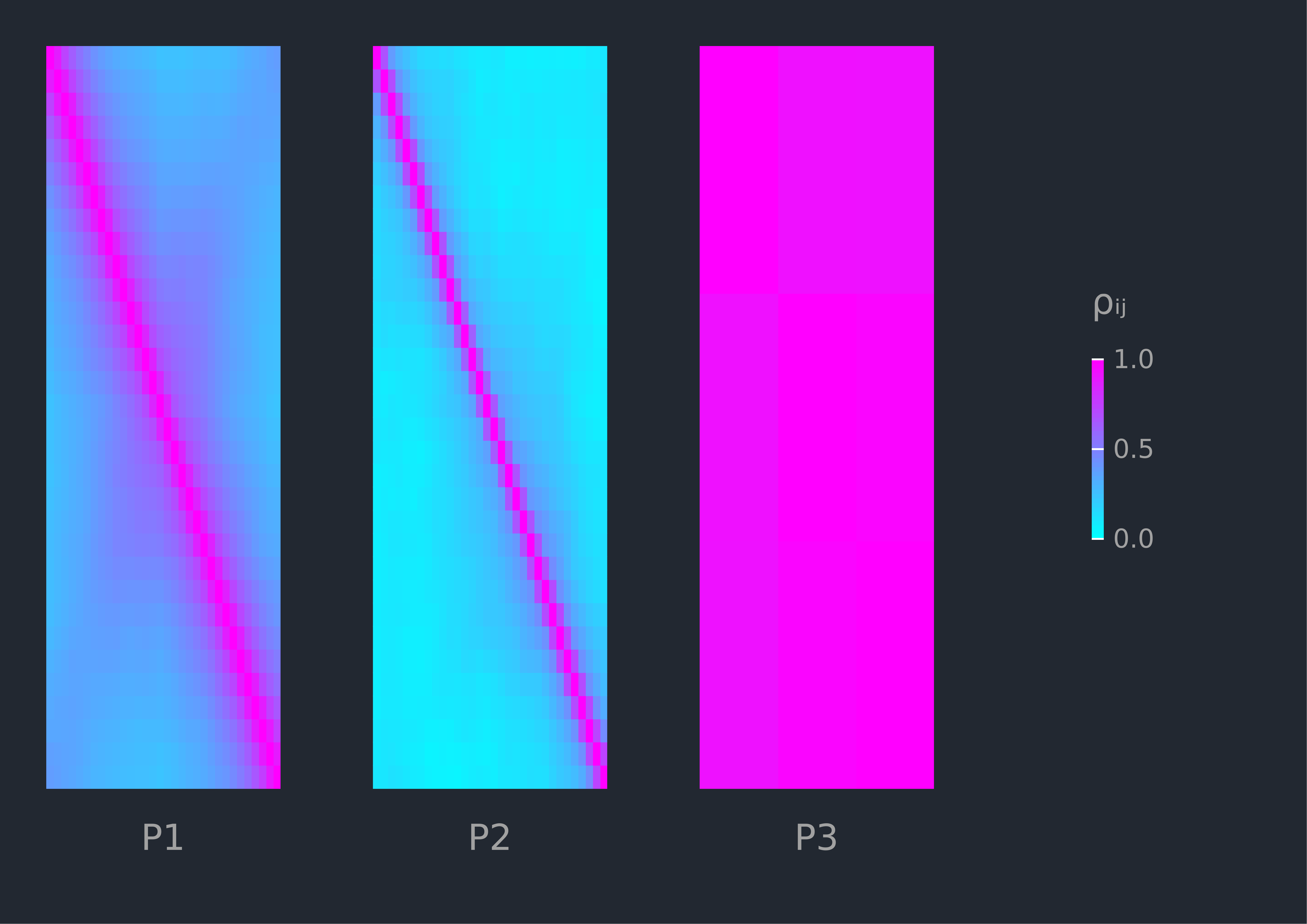}
		\caption{Correlation matrices, $\mathbf{P}_d$, from the NIG model for the image data.}
		\label{fig:img_cor}
	\end{center}
\end{figure}

We accessed the CIFAR-100 data through the {\tt MLDatasets.jl} package, version 0.5.3. We began with the training data and chose the images that corresponded to the class {\tt maple\_tree}. The maple tree images were converted from RGB arrays to an HSV format to filter out trees that did not have green or yellow leaves. 

The location tensor, $\tenM$ is visualized in figure \ref{fig:img_loc}. Like when looking at the mean, we still see that the image looks like a generic tree with a brown trunk and a blue sky.

Figure \ref{fig:img_cor} displays the three scale matrices, $\{\matdel_{d}\}_{d=1}^3$ as correlation matrices, $\{\mathbf{P}_{d}\}_{d=1}^3$. The correlation pattern for the rows and columns indicates entries close together are positively correlated and the correlation decreases for pixels that are further apart. This pattern is to be expected for image data. 

\end{document}